\documentclass[11pt,twoside]{article}

\usepackage{geometry}                
\usepackage{graphicx}
\usepackage{amsmath,amssymb, amsthm,epsfig,bm}
\usepackage{setspace}
\usepackage{epstopdf}
\usepackage{psfrag}
\usepackage{color,soul}
\usepackage{verbatim}
\usepackage{algorithm,algorithmic}
\usepackage{multirow}
\usepackage{natbib}
\usepackage{bibentry}
\usepackage{xr}
\usepackage{mathrsfs}
\usepackage{cases}
\usepackage[colorlinks=true,citecolor=blue,linkcolor=blue]{hyperref}
\usepackage{enumitem} 
\usepackage{subfig}
\usepackage{caption}

\newcommand{\titleshort}{Confidence sets by projection and shrinkage}
\newcommand{\authorshort}{Zhou, Li, and Zhou}

\setlist{noitemsep, topsep=0pt}
\bibpunct{(}{)}{;}{a}{}{,} 
\usepackage[english]{babel}
\usepackage[utf8]{inputenc}
\usepackage{fancyhdr}
 
\numberwithin{equation}{section}
\theoremstyle{plain}
\newtheorem{theorem}{Theorem}
\newtheorem{lemma}{Lemma}

\newtheorem{corollary}[theorem]{Corollary}

\theoremstyle{definition}

\theoremstyle{remark}
\newtheorem{remark}{Remark}

\newcommand{\bfI}{\mathbf{I}}

\newcommand{\scrB}{\mathscr{B}}
\newcommand{\scrC}{\mathscr{C}}

\newcommand{\R}{\mathbb{R}}

\newcommand{\Prob}{\mathbb{P}}

\newcommand{\alp}{\alpha}

\newcommand{\hbeta}{\hat{\beta}}

\newcommand{\tdbeta}{\tilde{\beta}}

\newcommand{\veps}{\varepsilon}

\newcommand{\hmu}{\hat{\mu}}

\newcommand{\hsigma}{\hat{\sigma}}

\newcommand{\eps}{\epsilon}
\newcommand{\lmd}{\lambda}

\newcommand{\Omg}{\Omega}

\newcommand{\hL}{\hat{L}}

\newcommand{\E}{\mathbb{E}}

\newcommand{\defi}{\mathop{:=}} 
  
\newcommand{\dnorm}{\mathcal{N}}

\newcommand{\trans}{\mathsf{T}}

\newcommand{\sqn}{\sqrt{n}}

\newcommand{\wh}{\widehat}

\providecommand{\norm}[1]{\lVert#1\rVert}

\DeclareMathOperator*{\argmin}{argmin}

\DeclareMathOperator{\rank}{rank}
\DeclareMathOperator{\supp}{supp}
\DeclareMathOperator{\spn}{span}

\DeclareMathOperator{\tr}{tr}

\RequirePackage{tikz}

\newcommand{\cset}{\wh{C}}
\newcommand{\Hset}{\mathcal{H}}
\graphicspath{{graph/}}

\onehalfspace

\begin{document}

\title{Honest Confidence Sets for High-Dimensional Regression by Projection and Shrinkage}

\author{Kun Zhou, Ker-Chau Li, and Qing Zhou\thanks{Department of Statistics, University of California, Los Angeles. Email: k.zhou@ucla.edu, kcli@stat.ucla.edu, zhou@stat.ucla.edu}}
\date{}
\maketitle

\begin{abstract}
The issue of honesty in constructing confidence sets arises in nonparametric regression.
While optimal rate in nonparametric estimation can be achieved and utilized to construct sharp confidence sets, severe degradation of confidence level often happens after estimating the degree of smoothness. Similarly, for high-dimensional regression, oracle inequalities for sparse estimators could be utilized to construct sharp confidence sets. Yet the degree of sparsity itself is unknown and needs to be estimated, causing the honesty problem. To resolve this issue, we develop a novel method to construct honest confidence sets for sparse high-dimensional linear regression. The key idea in our construction is to separate signals into a strong and a weak group, and then construct confidence sets for each group separately. This is achieved by a projection and shrinkage approach, the latter implemented via Stein estimation and the associated Stein unbiased risk estimate. Our confidence set is honest over the full parameter space without any sparsity constraints, while its diameter adapts to the optimal rate of $n^{-1/4}$ when the true parameter is indeed sparse. Through extensive numerical comparisons, we demonstrate that our method outperforms other competitors with big margins for finite samples, including oracle methods built upon the true sparsity of the underlying model.

{\em Keywords}: adaptive confidence set, high-dimensional inference, sparse linear regression, Stein estimate.
\end{abstract}

\section{Introduction}\label{sec:intro}
Consider high-dimensional linear regression
\begin{align}\label{highdimensionRL}
y = X \beta + \veps,
\end{align}
where $y \in \R^n$, $X = [X_1 | \cdots | X_p] \in \R^{n \times p}$, $\beta \in \R^p$, $\veps \sim \dnorm_n(0, \sigma^2\bfI_n)$ and $p > n$.
While there is a rich body of research on parameter estimation under this model concerning signal sparsity (e.g.\ \cite{Bickel09, Zhang08, negahban12}), how to construct confidence sets remains elusive.
In this work, we focus on  confidence sets for the mean $\mu=X\beta$ with the following two properties:
First, the confidence set $\wh{C}$ is (asymptotically) honest over all possible parameters.
That is, for a given confidence level $1-\alp$,
\begin{align}\label{eq:asymhonestdef}
\liminf_{n\to\infty}\inf_{\beta \in \R^p} \Prob_{\beta}\left\{ X\beta\in \wh{C} \right\}\geq 1-\alp,
\end{align}
where $\Prob_{\beta}$ is taken with respect to the distribution of $y\sim \dnorm_n(X\beta,\sigma^2\bfI_n)$, regarding $X$ as fixed.
Second, the diameter of $\wh C$ is able to adapt to the sparsity and the strength of $\beta$.
In practical applications, sparsity assumptions are very hard to verify, and for many data sets they are at most a good approximation.
The first property guarantees that our confidence sets reach the nominal coverage without imposing any sparsity assumption, while the second property allows us to leverage sparse estimation when $\beta$ is indeed sparse.

Our problem is related to the construction of confidence sets in nonparametric regression, for which a line of work has laid down important theoretic foundations and provided methods of construction \citep{Li89,Beran98,hoffman02,Juditsky03,Baraud04,genovese05,Robins06,cai06,Bull13}.
Despite such notable advances, lack of numerical support casts doubt on the merit of borrowing these nonparametric regression methods directly for sparse regression.
Taking the adaptive method based on sample splitting in \cite{Robins06} as an example, an honest confidence set for $\mu$ can be constructed as
$\wh C_a=\{\mu\in\R^n:n^{-1/2}\|\mu-X\hbeta\|\leq r_n\}$, where $X\hbeta$ is an initial estimate independent of $y$, and its (normalized) diameter
$|\wh C_a|\defi 2r_n=O_p(n^{-1/4}+n^{-1/2}\|X\hbeta-X\beta\|)$.
A common choice for $\hbeta$ under model~\eqref{highdimensionRL} for $p>n$ is a sparse estimator, such as the lasso \citep{tibshirani96} or $\ell_0$-penalized least-squares estimator.
With high probability, the prediction loss of the lasso estimator typically satisfies
\begin{align}\label{eq:sparseeb}
\frac{1}{n}\|X\hbeta-X\beta\|^2 \leq c \frac{s \log p}{n}
\end{align}
for some $c>0$, uniformly for all $\beta\in \scrB(s):=\{v\in\R^p: \|v\|_0\leq s\}$; see for example \cite{Bickel09}.
Under this choice, the diameter $|\wh C_a|$ is of the order
\begin{align}\label{eq:sparsediam}
|\wh C_a|=O_p\left(n^{-1/4}+\sqrt{s \log p/n}\right)
\end{align}
for all $\beta\in \scrB(s)$.
For a precise statement, see Theorem~\ref{thm:CRobins} below.
This method has nice theoretical properties when $s=o(n/\log p)$. But even for moderately sparse signals with $ s/n\to\delta\in (0,1)$, the bound on the right side of \eqref{eq:sparsediam} approaches $\infty$ as $p>n\to \infty$ and thus offers little insight into the performance of the confidence set. 
The upper bound \eqref{eq:sparseeb} also critically depends on the regularization parameter used for the initial estimate $\hbeta$.
In fact, our numerical results show that, for finite samples with $(s,n,p)=(10,200,800)$, this confidence set can be worse than a naive $\chi^2$ region $\{\mu: \|y-\mu\|^2\leq \sigma^2 \chi^2_{n,\alp}\}$, where $\chi^2_{n,\alp}$ denotes the $1-\alp$ quantile of the $\chi^2$ distribution with $n$ degrees of freedom.
A similar issue occurs in the related but different problem of constructing confidence sets for $\beta$.
\cite{Nickl13} have shown that one can construct a confidence set for $\beta$ that is honest over $\scrB(k_1)$ for $k_1=o(n/\log p)$, and for any $s\leq k_1$, the diameter is on the same order as that in \eqref{eq:sparsediam} for any $\beta\in\scrB(s)$. Compared to the unrestricted  honesty in \eqref{eq:asymhonestdef} over the entire space $\R^p$, the restriction on the honesty region to $\scrB(k_1)$ also reflects the challenge faced in the construction of confidence sets when $p>n$.
Recently, \cite{ewald18} provide an exact formula to compute a lower bound of the coverage rate of a confidence set centered at the lasso, over the entire parameter space for any significance level $\alpha \in (0, 1)$, and vice versa;
however, low dimension $(p<n)$ is a vital condition in their proof, making it impossible to generalize their idea to the high-dimensional problem that we are studying.

The construction of confidence sets is fundamentally different from the problem of inferring error bounds for a sparse estimator \citep{Nickl13}. 
It is seen from \eqref{eq:sparsediam} that no matter how sparse the true $\beta$ is, the diameter of $\wh C_a$ cannot converge at a rate faster than $n^{-1/4}$.
Indeed, results in \cite{Li89} imply that, for the linear model \eqref{highdimensionRL} with $p\geq n$, the diameter of an honest confidence set for $\mu$, in the sense of \eqref{eq:asymhonestdef}, cannot adapt at any rate $o(n^{-1/4})$.
This is in sharp contrast to error bounds for a sparse estimator, such as that in \eqref{eq:sparseeb}, which can decay at a much faster rate when $\beta$ is sufficiently sparse.
It is not desired to construct confidence sets directly from error bounds like \eqref{eq:sparseeb} even we only require honesty for $\beta\in\scrB(k_1)$ with a given $k_1=o(n/\log p)$, because its diameter, on the order of $\sqrt{k_1\log p /n}$, cannot adapt to any sparser $\beta\in\scrB(s)$ for $s< k_1$.

Motivated by these challenges, we propose a new two-step method to construct a confidence set for $\mu=X\beta$, allowing the dimension $p\gg n$ in \eqref{highdimensionRL}.
The basic idea of our method is to estimate the radius of the confidence set separately for strong and weak signals defined by the magnitude of $|\beta_j|$.
Using a sparse estimate, such as the lasso, one can recover the set $A$ of large $|\beta_j|$ accurately and expect a small radius for a confidence ball for $\mu_A$, the projection of $\mu$ onto the subspace spanned by $X_j, j\in A$.
By construction, $(\mu-\mu_A)$ is composed of weak signals.
Thus, in the second step, we shrink our estimate of this part towards zero by Stein's method and construct a confidence set with Stein's unbiased risk estimate \citep{Stein81}.
Combining the inferential advantages of sparse estimators and Stein estimators, our method overcomes many of the aforementioned difficulties.
First, our confidence set is honest for all $\beta\in\R^p$, and its diameter is well under control for all possible values of $\beta$ including the dense case.
Second, by using elastic radii our confidence set, an ellipsoid in general, can adapt to signal strength and sparsity.
The radius for strong signals adapts to the sparsity of the underlying model via sparse estimation or model selection, while the radius for weak signals adapts according to the degree of shrinkage of the Stein estimate.
Without any signal strength assumption, the diameter of our confidence set is $O_p(n^{-1/4} + \sqrt{s\log p / n})$, the same as \eqref{eq:sparsediam}, for $\beta\in\scrB(s)$. It may further reduce to $O_p(n^{-1/4} + \sqrt{s/n})$ under an assumption on the separability between the strong and the weak signals, which shrinks to the optimal rate $n^{-1/4}$ when the signal sparsity $s=O(\sqn)$, as opposed to $s=O(\sqn/\log p)$ in \eqref{eq:sparsediam}.
Third, we provide a data-driven selection of the set $A$ from multiple candidates, which protects our method from a bad choice and thus makes it very robust.
We demonstrate with extensive numerical results that our method can construct much smaller confidence sets than other competing methods, including the adaptive method \citep{Robins06} discussed above and oracle approaches making use of the {\em true\/} sparsity of $\beta$ (the oracle).
These results highlight the practical usefulness of our method.

Note that the construction of confidence sets for $\mu=X\beta$ is different in nature from the construction of confidence intervals for an individual $\beta_j$ or a low-dimensional projection of $\beta$.
For the latter, the optimal rate of an interval length can be $n^{-1/2}$ when $\beta$ is sufficiently sparse \citep{schneider16,Cai17}, such as the intervals constructed by de-biased lasso methods \citep{zhang14, Van14,javanmard14}. Although simultaneous inference methods have been proposed based on bootstrapping de-biased lasso estimates \citep{ZhangCheng17,Dezeure2017}, these methods are shown to achieve the desired coverage only for extremely sparse $\beta$ such that $\|\beta\|_0=o(\sqrt{n/(\log p)^3})$, which severely limits their practical application.


The remainder of this paper is organized as follows:
Section~\ref{sec:twostepstein} develops our two-step Stein method in details, including its theoretical properties and algorithmic implementation.
To demonstrate the advantage of our method, we develop in Section~\ref{sec:twosteplasso} a few competing methods making use of the lasso prediction or the oracle of the true sparsity.
Extensive numerical comparisons are provided in Section~\ref{sec:simulation} to show the superior performance of our two-step Stein method, relative to the competitors, in a variety of sparsity settings, including when $\beta$ is quite dense.
The paper is concluded in Section~\ref{sec:discussion} with further discussions.
Proofs of all theoretical results are deferred to the Supplementary Material.

Throughout the paper, we always assume model \eqref{highdimensionRL} with $\veps \sim \dnorm_n (0, \sigma^2\bfI_n)$ unless otherwise noted.
We denote by $\Prob_{\beta}$ the distribution of $[y\mid X]$ and $\E_{\beta}$ the corresponding expectations, where the subscript $\beta$ may be dropped when its meaning is clear from the context.
Denote by $[p]$ the index set $\{ 1, \ldots, p \}$ and by $|A|$ the size of a set $A \subseteq [p]$.
Write $a_n = \Omega(b_n)$ if $b_n = O(a_n)$ and $a_n \asymp b_n$ if $a_n = O(b_n)$ and $b_n = O(a_n)$.
We use $\Omega_p(.)$ and $\asymp_p$ if the above statements hold in probability.
For a vector $v = (v_j)_{1:m}$, let $v_A = (v_j)_{j \in A}$ be the restriction of $v$ to the components in $A$.
For a matrix $M=[M_1\mid \ldots \mid M_m]$, where $M_j$ is the $j$th column, denote by $M_A = (M_j)_{j \in A}$ the submatrix consisting of columns in $A$.
For $a,b \in \R^n$, $\langle a,b \rangle \mathop{:=} a^\trans b$ is the inner product.
Define $a \vee b \mathop{:=} \max\{a, b\}$ and $a \wedge b \mathop{:=} \min\{a, b\}$ for $a,b\in\R$.

\section{Two-step Stein method}\label{sec:twostepstein}

Dividing $\beta$ into strong and weak signals, our method constructs a confidence set $\wh{C}(y)$ with an ellipsoid shape for $X\beta$ that is honest as defined in \eqref{eq:asymhonestdef}.
Note that under a high-dimensional asymptotic framework, all variables $X=X(n)$, $y=y(n)$, $\beta=\beta(n)$ and $s=s_n$ depend on $n$ as $p=p_n\gg n \to \infty$, while
$X(n)$ is regarded as a fixed design matrix for each $n$. We often suppress the dependence on $n$ to simplify the notation.

\subsection{Preliminaries on Stein estimation}
We will use a simplified Stein estimate \citep{Li89} to construct the confidence set for weak signals.
For a linear estimate $\tilde{\mu} = T_n y$, where $y \sim \dnorm_n(\mu, \sigma^2 \bfI_n)$ and $T_n \in \R^{n \times n}$, let $R_n = \bfI_n - T_n,$ and define
\begin{align}
\hmu(y;\tilde{\mu}) &= y - \frac{\sigma^2 \tr(R_n)}{\|R_ny\|^2}R_ny, \label{eq:stein_mu}\\
\hat{L}(y;\tilde{\mu}) &= 1 - \frac{\sigma^2\left(\tr(R_n)\right)^2}{n\|R_ny\|^2}, \label{eq:stein_L}
\end{align}
where $\hmu(y;\tilde{\mu})$ is the Stein estimate associated with the initial estimate $\tilde{\mu}$ and $\sigma^2\hat{L}(y;\tilde{\mu})$ is the Stein unbiased risk estimate (SURE).
\cite{Li89} proved the uniform consistency of $\hat{L}$.
\begin{lemma}[Theorem~3.1 in \cite{Li89}]\label{thm:surebound}
Assume that $y \sim \dnorm_n(\mu,\sigma^2\bfI_n)$. For any $\alp\in(0,1)$, there exists a constant $c_s(\alp)>0$ such that
\begin{align}
\liminf_{n\to\infty} \inf_{\mu\in\R^n}
\Prob_{\mu}\left\{ \left|\sigma^2\hL - {n^{-1}}\|\hmu - \mu\|^2\right|\leq c_s(\alp) \sigma^2 n^{-1/2}\right\} \geq 1-\alp,\label{eq:surebound}
\end{align}
where $\hmu$ and $\hL$ are defined in \eqref{eq:stein_mu} and \eqref{eq:stein_L}.
\end{lemma}


\subsection{Method of construction}\label{sec:singlecandidate}

Now, consider the linear model \eqref{highdimensionRL} and let $\mu=X\beta$.
Given a pre-constructed candidate set $A = A_n \subseteq [p]$, independent of $(X,y)$, define
\begin{align*}
\mu_{A}=P_A \mu,\quad\quad\mu_\perp=P_A^\perp \mu=(\bfI_n-P_A)\mu,
\end{align*}
where $P_A$ is the orthogonal projection from $\R^n$ onto $\spn(X_A)$ and $P_A^\perp$ is the projection to the orthogonal complement.
A good candidate set $A$ is supposed to include all strong signals, say $A=\{j: |\beta_j|>\tau\}$. With such a choice, $\|\mu_\perp\|$ will be small.
Typically, we split our data set into two halves, $(X,y)$ and $(X',y')$, and apply a model selection method on $(X',y')$ to construct the set $A$.
See Section~\ref{sec:convergencerate} for more detailed discussion.

We estimate $\mu_{A}$ and $\mu_\perp$, respectively, by $\hmu_{A}$ and $\hmu_\perp$, compute radii $r_A$ and $r_\perp$, and construct a $(1-\alp)$ confidence set $\cset$ for $\mu$ in the form of
\begin{align}\label{eq:ellipsoidmu}
\cset=&\left\{\mu\in\R^n: \frac{\|P_A \mu-\hmu_{A}\|^2}{n r_A^2}+\frac{\|P_A^{\perp} \mu-\hmu_\perp\|^2}{n r_{\perp}^2} \leq 1\right\}.
\end{align}
Note that $\cset$ is an ellipsoid in $\R^n$, where $r_A=r_A(\alp)$ and $r_\perp=r_\perp(\alp)$ correspond to the major and minor axes, respectively.
Our method consists of a projection and a shrinkage step:

{\em Step 1: Projection.\/}
Let $\hmu_A=P_A y $ and $k = \rank(X_A) \leq |A|$.
Since $A$ is independent of $(y,X)$, we have 
\begin{align}
\|\hmu_{A}-\mu_{A}\|^2 = \|P_A \veps\|^2 \mid A \sim \sigma^2 \chi^2_{k}.\label{eq:hmu_A}
\end{align}
Thus, we choose
\begin{align}\label{eq:ra}
r_A^2=c_1 \tilde{r}_A^2 = c_1 \sigma^2\chi^2_{k,\alp/2}/n,
\end{align}
where $\chi^2_{k,\alp/2}$ is the $(1-\alp/2)$ quantile of the $\chi^2_{k}$ distribution and $c_1 > 1$ is a constant, so that
\begin{align}\label{eq:coverageA}
\Prob\left\{ \frac{\|P_A \mu-\hmu_{A}\|^2}{n r_A^2}\leq 1/c_1\right\} = 1-\alp/2.
\end{align}

{\em Step 2: Shrinkage.\/}
Let $y_\perp=P_A^\perp y$.
As mentioned above, under a good choice of $A$ that contains strong signals, $\|\mu_\perp\|$ is expected to be small.
Therefore, we shrink $y_\perp$ towards zero via Stein estimation to construct $\hmu_\perp$.
Note that $y_\perp$ is in an $(n-k)$-dimensional subspace of $\R^n$.
Letting $\tilde{\mu} = 0$ and $R_n = P_A^\perp$ in \eqref{eq:stein_mu} and \eqref{eq:stein_L}, we obtain
\begin{align}
 \hmu_\perp &= \hmu(y_\perp;0)=(1-B)y_\perp, \label{steinestimate}\\
 \hL &= \hat{L}(y_\perp;0) =(1-B), \label{steinL}
\end{align}
where the shrinkage factor
\begin{align} \label{steinB}
B =(n-k)\sigma^2/\|y_\perp\|^2.
\end{align}
It then follows from Lemma~\ref{thm:surebound} that
\begin{align} \label{eq:perpsurebound}
\liminf_{(n - k)\to\infty} \inf_{\beta\in\R^p}
\Prob\left\{ \left|\sigma^2 \hL-{(n - k)^{-1}}\|\hmu_\perp-\mu_\perp\|^2\right|\leq c_s(\alp) \sigma^2 (n - k)^{-1/2}\right\} 
 \geq 1-\alp, 
\end{align}
for any sequence of $A=A_n$ as long as $(n - k)\to\infty$.
Therefore, if we choose
\begin{align}\label{eq:rperp}
r_\perp^2= c_2 \tilde{r}_\perp^2 = c_2 \frac{n-k}{n}\sigma^2\left\{\hL + c_s(\alp/2)(n-k)^{-1/2}\right\},
\end{align}
where $c_2 > 1$ is a constant, we have
\begin{align}\label{eq:coverageperp}
\liminf_{(n-k)\to\infty} \inf_{\beta\in\R^p}
\Prob\left\{ \frac{\| \mu_\perp - \hmu_\perp\|^2}{n r_{\perp}^2}\leq 1/c_2\right\} \geq 1-\alp/2.
\end{align}
In practical implementation, we estimate the constant $c_s(\alp)$ in \eqref{eq:perpsurebound} by simulation, which will be discussed in Section~\ref{sec:algorithmstein}.

If $1 / c_1 + 1 / c_2 = 1$, confidence set \eqref{eq:ellipsoidmu} made up from \eqref{eq:coverageA} and \eqref{eq:coverageperp} is honest and the expectation of its (normalized) diameter $|\cset| \mathop{:=} 2(r_A \vee r_\perp)$ can be calculated explicitly for all $\beta\in\R^p$:
\begin{theorem}\label{thm:twostephonest}
Assume $1 / c_1 + 1 / c_2 = 1$, $A$ is independent of $(y,X)$ with $\rank(X_A)=k$, and $(n-k)\to\infty$ as $n\to\infty$.
Then the confidence set $\wh C$ \eqref{eq:ellipsoidmu} constructed by the two-step Stein method is honest in the sense of \eqref{eq:asymhonestdef}. Furthermore, the squared diameter of $\wh C$ has expectation
\begin{align}\label{eq:expection_diameter}
\E|\wh C|^2 =  & 4\sigma^2\max\left\{c_1\frac{\chi^2_{k,\alp/2}}{n}, c_2\frac{n-k}{n}\left(1 - \E\frac{n-k}{\chi^2_{n-k}(\rho)} + c_s(\alp/2)(n-k)^{-1/2} \right)\right\},
\end{align}
where $\chi^2_{n-k}(\rho)$ follows a noncentral $\chi^2$ distribution with $n-k$ degrees of freedom and non-centrality parameter $\rho = \|\mu_\perp\|^2/\sigma^2$.
\end{theorem}

In the above result, we did not impose any assumptions on $A$ except $(n - k) \to \infty$, which allows many choices of $A$.
Our confidence set $\wh C$ is honest as in \eqref{eq:asymhonestdef} and its diameter is under control for all $\beta\in\R^p$.
Since $\E[1/\chi^2_{n-k}(\rho)]>0$, a uniform but very loose upper bound
\begin{align}\label{eq:univbound}
\E|\wh C|^2 \leq 4\sigma^2\max\left\{c_1\frac{\chi^2_{k,\alp/2}}{n}, c_2\frac{n-k}{n}\left(1 + c_s(\alp/2)(n-k)^{-1/2}\right)\right\}
\end{align}
holds for all $\beta\in\R^p$.
In particular, when $\beta$ is dense, the diameter will be comparable to that of the naive $\chi^2$ region.
As corroborated with the numerical results in Section~\ref{subsec:dense}, this protects our method from inferior performance when sparsity assumptions are violated, making it robust to different data sets.
Next, we will show that our confidence set is adaptive: When $\beta$ is indeed sparse with separable strong and weak signals, the radii $r_A$ and $r_\perp$ will adapt to the optimal rate with a proper choice of $A$ that contains strong signals.

\subsection{Adaptation of the diameter}\label{sec:convergencerate}
To simplify our analysis, we set $c_1= c_2 = 2$ in this section so that they can be ignored when calculating the convergence rates of $r_A$ and $r_\perp$.
These rates do not change as long as $c_1$ and $c_2$ stay as constants when $n\to \infty$.
Lemma~\ref{lm:radiusrate} specifies conditions for the diameter of $\wh C$ to converge at the optimal rate $n^{-1/4}$.

\begin{lemma}\label{lm:radiusrate}
Suppose that $k = \rank(X_A)$ and $\norm{\mu_\perp}=o(\sqrt{n-k})$. Then
\begin{align*}
r_A^2 \asymp_p {k/n}, \quad 
r_\perp^2 = O_p \left(\frac{\sqrt{n-k}}{n} + \frac{\norm{\mu_\perp}^2}{n}\right).
\end{align*}
Therefore, if $k=O(\sqn)$ and $\|\mu_\perp\| = O(n^{1/4})$, then the diameter of $\cset$ 
\begin{align*}
|\cset|=2(r_A \vee r_\perp) \asymp_p n^{-1/4}.
\end{align*}
\end{lemma}

The $\ell_2$ norm of the weak signals $\norm{\mu_\perp}$ can be bounded by $\| \beta_{A^c} \|$ under the sparse Riesz condition on $X$ and a sparsity assumption on $\beta$.
A design matrix $X$ satisfies the sparse Riesz condition \citep{Zhang08} with rank $s^*$ and spectrum bounds $0 < c_* < c^* < \infty$, denoted by $\textup{SRC}(s^*, c_*, c^*)$, if
\begin{align*}
c_* \leq \frac{\norm{X_Av}^2}{n\norm{v}^2} \leq c^*, \quad \text{for all } A \text{ with }|A| = s^* \text{ and all nonzero } v \in \R^{s^*}.
\end{align*}
Under our asymptotic framework, $s^*$, $c^*$ and $c_*$ are allowed to depend on $n$.

\begin{theorem}\label{thm:optimalA}
Suppose $X$ satisfies $\textup{SRC}(s^*, c_*, c^*)$ with $s^* \geq |\supp(\beta) \cap A^c|$, and let $k=\rank(X_A)$.
If $\limsup_{n} c^* < \infty$, $k=o(n)$ and $\norm{\beta_{A^c}}=o(1)$, then
\begin{align}\label{eq:steinsparserate}
|\cset| =O_p\left\{ (n^{-1/4} + \norm{\beta_{A^c}}) \vee \sqrt{k/n}  \right\}
\end{align}
for the two-step Stein method.
In particular, $|\cset| \asymp_p n^{-1/4}$ if $k=O(\sqn)$ and $\norm{\beta_{A^c}} = O(n^{-1/4})$.
\end{theorem}

\begin{remark}\label{rm:steinopt}
Let us take a closer look at the conditions in this theorem for $|\cset| \asymp_p n^{-1/4}$.
Suppose that $\beta$ has $O(\sqn)$ strong coefficients that can be reliably detected by a model selection method, while all other signals are weak such that $\norm{\beta_{A^c}}=O(n^{-1/4})$. Then we can have {$k \leq |A|=O(\sqn)$} with high probability.
This shows that the sparsity $s=\|\beta\|_0$ is allowed to be $O(\sqn)$.
The only additional constraint on $s$ comes from the assumption $\textup{SRC}(s^*, c_*, c^*)$ with $s^* \geq s$, which holds for Gaussian designs if $s\log p =o(n)$ \citep{Zhang08}. Compared to \eqref{eq:sparsediam} which requires $s\log p =O(\sqn)$, we have relaxed the sparsity assumption on $\beta$ to attain the optimal rate $n^{-1/4}$ by imposing a mild condition on the decay rate of the weak signals $\| \beta_{A^c} \|$. 
\end{remark}

Now we discuss a few methods to find $A$ so that our confidence sets can adapt to the sparsity and signal strength of $\beta$.
We split the whole data set into $(X, y)$ and $(X', y')$, with respective sample sizes $n$ and $n'$, so that they are independent.
Henceforth, we assume an even partition with $n'=n$, which simplifies the notation and is commonly used in practice, unless otherwise noted.
The first method is to apply lasso on $(X', y')$:
\begin{align}\label{eq:lassodef}
\hbeta = \hbeta(y',X';\lambda)\defi\argmin_{\beta\in\R^p} \left[\frac{1}{2n} \|y' - X'\beta\|^2 +\lmd \|\beta\|_1\right],
\end{align}
where $\lmd$ is a tuning parameter.
Then choose
\begin{align}
A = \{j : \hbeta_j \neq 0 \} \label{eq:lassothresh},
\end{align}
that is, we define strong signals by the support of the lasso.
This choice of $A$ is justified by the following corollary.
Let $A_0 = \supp(\beta)$ and $S_0 = \{ j \in A_0: |\beta_j| \geq K \sqrt{s\log p / n} \}$ for a sufficiently large $K$.


\begin{corollary}\label{cor:lassodi}
  Suppose that $X$ and $X'$ satisfy $\textup{SRC}(s^*, c_*,c^*)$, where $0<c_*<c^*$ are constants.
  Let the confidence set $\cset$ \eqref{eq:ellipsoidmu} be constructed by the two-step Stein method with $A$ chosen by \eqref{eq:lassothresh} and $\lambda = c_0\sigma \sqrt{c^*\log p / n}, c_0 > 2\sqrt{2}$.
  Assume $s \leq (s^* - 1) / (2 + 4c^*/c_*)$ and $s\log p = o(n)$. Then for any $\beta\in\scrB(s)$ we have
  \begin{align}\label{eq:lassosteinradius}
  |\cset|=O_p\left(n^{-1/4}+ \sqrt{s\log p / n}\right).
  \end{align}
If in addition $\norm{\beta_{A_0\setminus S_0}}=O(n^{-1/4})$, then 
  \begin{align}\label{eq:lassosteinradiussep}
  |\cset|=O_p\left(n^{-1/4}\vee \sqrt{s/n}\right).
  \end{align}
\end{corollary}

The rate of $|\cset|$ in \eqref{eq:lassosteinradius} does not depend on any assumption on signal strength, and it is identical to \eqref{eq:sparsediam}. However, our method can achieve a faster rate \eqref{eq:lassosteinradiussep} if $\norm{\beta_{A_0\setminus S_0}}=O(n^{-1/4})$. Together with the definition of $S_0$, this essentially imposes a separability assumption between the strong and the weak signals when $s\log p \gg \sqrt{n}$. 


To weaken the beta-min condition on strong signals in $S_0$, we may apply a better model selection method to define $A$, such as using the minimax concave penalty (MCP) \citep{zhang10}:
\begin{align}
\rho(t; \lambda,\gamma)= \int_0^{|t|} \left(1 - \frac{u}{\gamma\lmd}\right)_+ \, du = 
\begin{cases}
   |t| - t^2 / (2\gamma \lambda) & \text{if}\;|t| \leq \gamma \lambda \\
   \gamma \lambda / 2 & \text{if}\;|t| > \gamma \lambda \\
\end{cases},
\end{align} 
for $\gamma > 1$. Accordingly, a regularized least-squares estimate is defined by
\begin{align}\label{eq:mcpdef}
  \hbeta^{\textup{mcp}}_{\lambda, \gamma} = \hbeta^{\textup{mcp}}_{\lambda, \gamma} (y',X') \defi
  \argmin_{\beta\in\R^p} \left[ 
  \frac{1}{2n} \|y' - X'\beta\|^2 +
  \lmd \sum_{j=1}^{p} \rho(|\beta_j|; \lambda, \gamma)
  \right].
\end{align}
Suppose we choose $A=\supp (\hbeta^{\textup{mcp}}_{\lambda, \gamma})$ in our two-step Stein method. The 
model selection consistency of $\hbeta^{\textup{mcp}}_{\lambda, \gamma}$ makes it possible for $|\cset|$ to adapt at the rate~\eqref{eq:lassosteinradiussep} under the same SRC assumption but a weaker beta-min condition than Corollary~\ref{cor:lassodi}. 

\begin{corollary} \label{cor:mcpsmp}
Suppose that $X$ and $X'$ satisfy $\textup{SRC}(s^*, c_*, c^*)$, where $0<c_*<c^*$ are constants, $s^* \geq (c^*/c_* + 1/2)s$, and $s\log p =o(n)$. Choose a sequence of $(\lambda_n, \gamma_n)$ satisfying $\lambda_n\gg \sqrt{\log p/ n}$ and $\gamma_n \geq c_*^{-1}\sqrt{4 + c_* / c^*}$. If $\beta\in\scrB(s)$ and $\inf_{A_0} |\beta_j | \geq (\gamma_n+1)\lambda_n$, then $\Prob\{\supp(\hbeta^{\textup{mcp}}_{\lambda_n,\gamma_n})=A_0\}\to 1$, and consequently the $\cset$ constructed by the two-step Stein method with $A=\supp(\hbeta^{\textup{mcp}}_{\lambda_n,\gamma_n})$ has diameter
\begin{align}\label{eq:mcprate}
  |\cset|=O_p\left(n^{-1/4}\vee \sqrt{s/n} \right).
  \end{align}
\end{corollary}

\begin{remark}\label{rmk:diameterStein}
Compared to \eqref{eq:sparsediam} for confidence sets centering at a sparse estimator, the diameter of our method in \eqref{eq:lassosteinradiussep} and \eqref{eq:mcprate} converges faster by a factor of $(\log p)^{1/2}$ when $s=\Omg(\sqn)$.
Accordingly, our method achieves the optimal rate when $s=O(\sqn)$ instead of $s=O(\sqn/\log p)$ as for \eqref{eq:sparsediam}.
Under a high-dimensional setting with $p\gg n$, say $p= \exp(n^a)$ for $a\in (0,1/2)$, this improvement in rate can be very substantial, which is supported by our numerical results. The faster rate of our method is made possible by its adaption to {\em both} signal strength and sparsity, while the rate of \eqref{eq:sparsediam} is obtained by adaption to sparsity only (cf. Theorem~\ref{thm:CRobins}). We emphasize that our method achieves the adaptive rates in the above results, while being uniformly honest over the entire $\R^p$ (Theorem~\ref{thm:twostephonest}). One could construct a confidence set with diameter $O_p(\sqrt{s/n})$ using only the covariates selected by a consistent model selection method, which would be faster than the rate \eqref{eq:mcprate}. However, such a confidence set is {\em not} honest over $\R^p$, because it cannot reach the nominal coverage rate for those $\beta$ that do not satisfy the required beta-min condition for model selection consistency. Our method overcomes this difficulty with the shrinkage step, based on the uniform consistency of the SURE (Lemma~\ref{thm:surebound}).
\end{remark}

\begin{remark}
For an uneven partition of the whole data set, the conclusions of Corollaries~\ref{cor:lassodi} and \ref{cor:mcpsmp} still hold as long as both $n' \asymp n \to \infty$. However, it is a common and reasonable choice to have $n=n'$, since $(X', y')$ and $(X, y)$ can be swapped to construct a confidence set for $X'\beta$, making full use of the whole data set.
\end{remark}

\subsection{Multiple candidate sets}\label{sec:mulitplecandidates}

It is common to have multiple choices for the candidate set $A$ in our two-step Stein method.
Let
\begin{align*}
\Hset = \{A_m \subseteq [p],\, m=1,\ldots,M_n\}
\end{align*}
be a collection of candidate sets.
We can apply the two-step Stein method to construct $M=M_n$ confidence sets for $\mu$, denoted by $\cset_m$, and then choose an optimal set $\cset_{m^*}$ by certain criterion such as minimizing the volume or the diameter.
Furthermore, the cardinality of $\Hset$ may be unbounded as $n$ increases, i.e., $M_n \to \infty$.
In what follows, we show that under mild conditions, \eqref{eq:coverageA} and \eqref{eq:coverageperp} hold uniformly for all $A\in\Hset$ after modifying $r_A$ and $r_\perp$ accordingly, which implies $\cset_{m^*}$ is asymptotically honest.

Put $k=\rank(X_A)$ for $A \in \Hset$ and $k_{\max} = \max_{A \in \Hset} k$.
Intuitively, the cardinality of $\Hset$ (i.e. $M$) and the maximum size of $A$ in $\Hset$ (i.e. $k_{\max}$) determine the radii and the coverage probability of $\wh C_m$.

For strong signals, we apply the following concentration inequality to show \eqref{eq:coverageA} holds uniformly:
\begin{lemma}\label{lemma:consistent_chi}
Suppose $\chi^2_n$ follows a $\chi^2$ distribution with $n$ degrees of freedom. Then for any $\delta > 0$,
\begin{align}\label{eq:consistency_chisq}
\Prob \left\{ \sqrt{n}\left| 1 - \frac{1}{n} \chi^2_{n} \right| \geq \delta \right\}
\leq
2\exp\left(-\frac{\delta^2}{4}\right).
\end{align}
\end{lemma}

\noindent
This lemma with a union bound implies
\begin{align*}
\Prob \left\{ \sup_{A \in \Hset } \sqrt{k} \left| \frac{\chi^2_k}{k} - 1 \right| \geq \delta \right\}
\leq \sum_{A \in \Hset } \Prob \left\{ \sqrt{k} \left| \frac{\chi^2_k}{k} - 1 \right| \geq \delta \right\}
\leq 2M\exp\left(-\frac{\delta^2}{4}\right).
\end{align*}
Then choosing
\begin{align}
r_A^2 = c_1 \tilde{r}_A^2 = \frac{c_1 \sigma^2}{n} \left[k + 2\sqrt{k \log (4M/\alpha)}\right] \label{eq:ra_consistency}
\end{align} 
as the radius for strong signals, we have
\begin{align*}
\Prob\left\{ \sup_{A \in \Hset } \frac{\|P_A \mu-\hmu_{A}\|^2}{n r_A^2}\leq 1/c_1\right\} \geq 1-\alp/2.
\end{align*}

For weak signals, we establish \eqref{eq:coverageperp} uniformly over $\Hset$ via the following result:

\begin{lemma}\label{thm:consistency}
Suppose all components of $\veps$ in \eqref{highdimensionRL}, $\veps_i, i=1, \ldots, n$, have mean $0$, common second, forth and sixth moments and their eighth moments are bounded by some constant $d$. For any $\delta > 0$ there exists a positive number $D$ depending on $d$ such that
\begin{align}
& \Prob \left\{ \sup_{A \in \Hset } \sqrt{n-k} \left|\sigma^2 \hL-{(n - k)^{-1}}\|\hmu_\perp-\mu_\perp\|^2\right| \geq \sigma^2\delta \right\}  \nonumber \\
& \quad
\leq
\Prob \left\{ \sup_{A \in \Hset} \sqrt{n-k} \left| \sigma^2 - \frac{1}{n - k} \| P_A^\perp \veps \|^2 \right| \geq \sigma^2\frac{\delta}{2} \right\}
+ D \sum_{A \in \Hset} \frac{1}{(n - k)^2}
+ D \frac{M}{\delta^4}. \label{eq:consistency} 
\end{align}
\end{lemma}


The proof of Lemma~\ref{thm:consistency} mainly follows the ideas in \cite{Li85}.
In our model with $\veps \sim \dnorm_n(0, \sigma^2 \bfI_n)$, the first term on the right hand side of \eqref{eq:consistency} simplifies to
\begin{align*}
\Prob \left\{ \sup_{A \in \Hset } \sqrt{n-k} \left| \sigma^2 - \frac{1}{n - k} \| P_A^\perp \veps \|^2 \right| \geq \sigma^2\frac{\delta}{2} \right\}
\leq
2M\exp\left(-\frac{\delta^2}{16}\right)
\end{align*}
via Lemma~\ref{lemma:consistent_chi}.
Assume that the cardinality of $\Hset$ and the maximum size of $A\in\Hset$ satisfy $M\ll (n - k_{\max})^2$.
To achieve the desired coverage for weak signals, it is sufficient to pick $\delta$ such that
$\delta^2 =\Omg(\log M)$ and $\delta^4=\Omg(M)$.
Therefore, we can set
$$\delta = c_{m}(\alp/2) M^{1/4} \gg (\log M)^{1/2}$$
for some constant $c_{m}(\alp/2)>0$, and the corresponding radius
\begin{align}\label{eq:rperp_consistency}
r^2_\perp = c_2 \tilde{r}^2_\perp = c_2 \frac{n - k}{n} \sigma^2 \left\{ \hL +  c_{m}(\alp / 2) \frac{M^{1/4}}{\sqrt{n - k}} \right\}
\end{align}
for any $A\in \Hset$, so that the upper bound in \eqref{eq:consistency} is $\leq \alp/2$.
Now we generalize Theorem~\ref{thm:twostephonest} to establish asymptotic honesty uniformly over $\Hset$:

\begin{theorem}
Given $\Hset$, construct confidence sets $\wh C_m, m=1,\ldots,M$, with $r_A$ and $r_\perp$ as in \eqref{eq:ra_consistency} and \eqref{eq:rperp_consistency}, respectively, for $A=A_m$.
Suppose $\lim_{n \to \infty} M / (n - k_{\max})^2 = 0$, $1/c_1+1/c_2=1$, and each $A_m$ is independent of $(X, y)$.
Then the confidence sets $\wh C_m$ are uniformly  honest over $\Hset$, i.e.,
\begin{align*}
\liminf_{n\to\infty} \inf_{\beta\in\R^p}\Prob\left[\bigcap_{m} \left\{X\beta \in \wh C_m\right\}\right]\geq 1-\alp.
\end{align*}
Consequently, $\cset_{m^*}$ chosen by any criterion is asymptotically honest.
\end{theorem}

\begin{remark}
The increment of $r_A^2$ in \eqref{eq:ra_consistency}, $2\sqrt{k \log (4M/\alpha)}/n$, reflects the cost for achieving uniform honesty over $\Hset$.
But this factor will not cause a slower rate for $r_A$ if $\log M=O_p(k)$, 
where the $k$ here is the size of the selected candidate set $A_{m^*}$.
Compared with \eqref{eq:rperp}, the factor $M^{1/4}/\sqrt{n-k}$ in \eqref{eq:rperp_consistency}, also the cost for uniform honesty, will in general lead to slower convergence of $r_\perp$.
However, this is a worthwhile price to protect our method from an improper candidate set $A$ that does not satisfy the assumptions in Theorem~\ref{thm:optimalA}.
For example, if the candidate set $A$ misses some strong signals, we may end up with $\hL \asymp_p 1$ and the radius of weak signals $r_{\perp}$ will not converge to $0$ at all.
Such bad choices of $A$ will be excluded if $\cset_{m^*}$ is chosen by minimizing its volume over $\Hset$.
In this sense, our method provides a data-driven selection of an optimal candidate set.
\end{remark}

To construct $\Hset$, we threshold the lasso $\hbeta$ in \eqref{eq:lassodef} calculated from $(X', y')$ to obtain
\begin{align}\label{eq:threshold}
A_m =\{j\in[p]: |\hbeta_j|>\tau_m\},
\end{align}
for a sequence of threshold values $\tau_m=a_m \lmd$, e.g. $a_m\in[0,4]$.
It is possible for two different $\tau_m$ to define the same $A$, which will be counted once in $\Hset$.
By setting $\tau_m = 0$ for some $m$, $A=\supp(\hbeta)$ will be included in $\Hset$, though it may not be selected as the optimal $\cset_{m^*}$.
In the proof of Corollary~\ref{cor:lassodi}, we have shown $\norm{\hbeta}_0 = O_p(\sqrt{n})$, and therefore both $M$ and $k_{\max}$ are $ O_p(\sqrt{n})$, which means
$M \ll (n - k_{\max})^2$ with high probability.
As a result, we can guarantee uniform honesty over all $\cset_m$.
Other choices of $\Hset$ are possible, such as stepwise variable selection with BIC.
It is possible that $A=\varnothing$ for a large value of $\tau_m$. In this special case, $r_A=0$, so the confidence set reduces to a ball, i.e.,
$\left\{\mu\in\R^n: {\|\mu-\hmu_\perp\|^2}\leq n r_{\perp}^2\right\}$.

\subsection{Algorithm and implementation}\label{sec:algorithmstein}
We implement our method with a sequence of candidate sets $A_m$ defined by \eqref{eq:threshold}.
Given the data set, $\sigma^2$, $\lambda$ in \eqref{eq:lassodef} and threshold values $\{a_m\lambda\}_{1 \leq m \leq M}$, this section describes some technique details in our algorithm to construct the confidence set \eqref{eq:ellipsoidmu} by the two-step Stein method.

{\em Data splitting.\/}
We split the original data set into $(X', y')$ and $(X, y)$.
Apply lasso on $(X', y')$ to get $\hbeta$ in \eqref{eq:lassodef} with the tuning parameter $\lambda$.
Threshold $\hbeta$ by $\tau_m = a_m\lambda$ for $m = 1, \ldots, M$ in \eqref{eq:threshold} to define candidate sets $A_m$.
Note that $A_m$, $m=1,\ldots, M$, are independent of $(X, y)$.

{\em Choice of $c_1$ and $c_2$.\/}
When $A \neq \varnothing$, we consider two criteria to choose the constants $c_1$ in \eqref{eq:ra} and $c_2$ in \eqref{eq:rperp}.
The first criterion is to minimize the log-volume of $\cset$, namely,
\begin{align*}
\log V(\cset) = k\log (r_{A}) + (n-k) \log (r_{\perp})
\end{align*}
up to an additive constant, which becomes a constrained optimization problem
\begin{align}\label{eq:logVolOptimization}
&\min_{c_1, c_2} \left\{ k\log (\sqrt{c_1}\tilde{r}_{A}) + (n-k) \log (\sqrt{c_2}\tilde{r}_{\perp}) \right\}, \\
&\text{ subject to } 1/ c_1 + 1 / c_2 = 1 \text{ and } 1 <c_1,c_2 \leq E, \nonumber
\end{align}
where $\tilde{r}_A$ and $\tilde{r}_\perp$ are defined in \eqref{eq:ra} and \eqref{eq:rperp} and $E>2$ is a pre-determined upper bound.
It is easy to obtain the solution
\begin{align} \label{eq:c1_c2_vol}
c_1 = \frac{E}{E - 1} \vee \left(\frac{n}{k} \wedge E\right), \qquad
c_2 = \frac{E}{E - 1} \vee \left(\frac{n}{n-k} \wedge E\right).
\end{align}
For all numerical results in this paper, we use $E=10$.
Without the constraint $c_1, c_2 \leq E$, the minimizer would be $(c_1, c_2) = (n/k, n/(n-k))$ so that under the conditions of Corollary~\ref{cor:lassodi},
$r_A = \sqrt{n / k} \tilde{r}_A \asymp_p 1$
and thus the diameter $|\cset|$ would not converge to $0$.
Therefore, a finite upper bound $E$ must be imposed.

The second criterion is to minimize the diameter $|\cset|$
\begin{align}\label{eq:volMinmax}
\min_{c_1, c_2} \max \{ r_A, r_\perp \}, \text{ subject to } 1/ c_1 + 1 / c_2 = 1,
\end{align}
which yields the solution
\begin{align} \label{eq:c1_c2_dia}
c_1 = (\tilde{r}_A^2 + \tilde{r}_\perp^2) / \tilde{r}_A^2, \qquad c_2 = (\tilde{r}_A^2 + \tilde{r}_\perp^2) / \tilde{r}_\perp^2.
\end{align}
As a result, we have
$r_A = r_\perp  = (\tilde{r}_A^2 + \tilde{r}_\perp^2)^{1/2}$ and the confidence set reduces to a ball.

{\em Computation of $c_s(\alpha)$.\/}
For any candidate set $A$, the radius $r_\perp$ \eqref{eq:rperp} depends on the constant $c_s(\alpha)$, which is essentially the quantile of the deviation between $\sigma^2\hL$ and the loss of the Stein estimator $\hmu_\perp$.
We use the following simulation procedure to estimate $c_s(\alpha)$:
First draw $\check{Y}_j \sim \dnorm_n(0, \sigma^2 \bfI_{n})$ for $j=1,2,\ldots,N$.
For each $j$, compute
\begin{align}
\check{\mu}_j=\left(1 - \frac{n\sigma^2}{\| \check{Y}_j\|^2}\right) \check{Y}_j \quad \text{and} \quad \check{L}_j =\left(1 - \frac{n\sigma^2}{\|\check{Y}_j\|^2}\right)_+. \label{eq:cAlphaSimulation}
\end{align}
Then the $(1 - \alpha)$ quantile of the empirical distribution of
\begin{align}
\frac{\sqrt{n}}{\sigma^2}\left| \sigma^2 \check{L}_j - n^{-1}{\|  \check{\mu}_j \|^2}\right|, \quad j=1,\ldots,N, \label{eq:cs_simulation}
\end{align}
is a consistent estimator of $c_s(\alpha)$ as long as $\|\mu_\perp\|=o(\sqn)$, which is the case under the assumptions of Corollary~\ref{cor:lassodi}.
Expression \eqref{eq:cs_simulation} can be written as a function of a $\chi^2_n$ random variable, which simplifies its simulation.

Clearly, the estimate of $c_s(\alpha)$ does not depend on $A$ and is used for any candidate set $A\in \Hset$ in our implementation.
Moreover, we find the multiple set adjustments on the radii, i.e., the factors of $(\log M)^{1/2}$ and $M^{1/4}$, are usually negligible given a reasonable sample size, say $n\geq 100$.
Therefore, we simply use the radii $r_A$ and $r_\perp$ in \eqref{eq:ra} and \eqref{eq:rperp} for each $A\in\Hset$.

Algorithm~\ref{alg:stein} summarizes the two-step Stein method with multiple candidate sets $A_m$.

\begin{algorithm}[ht]
\caption{Two-step Stein method\label{alg:stein}}
\begin{algorithmic}
\FOR{$m = 1, \ldots, M$}
    \STATE $A=A_m$
    \STATE compute $\hmu_{A}=P_A y$ and $\hmu_{\perp}$ by \eqref{steinestimate}
    \STATE compute $c_1$ and $c_2$ according to one of the two criteria
    \STATE compute $r_A$ and $r_\perp$ by \eqref{eq:ra} and \eqref{eq:rperp}
    \STATE construct $\wh C_m$ in the form of \eqref{eq:ellipsoidmu}
\ENDFOR
\STATE find $m^*$ by minimizing the volume or the diameter of $\cset_{m}$ over $m$
\end{algorithmic}
\end{algorithm}

\begin{remark}
In the calculation of $r_\perp$ and $c_s(\alpha)$, we use truncated SURE for $\hat{L}=(1-B)_+$ in \eqref{steinL} and similarly for $\check{L}_j$ in \eqref{eq:cAlphaSimulation}. Such a truncated rule has been used for the James-Stein estimator \citep{Efron73} and does not affect the asymptotic validity of our method.
\end{remark}

For all numerical results in this paper, we assume the noise variance $\sigma^2$ is known.
In real applications, one may use sample splitting to estimate $\hsigma=\hsigma(y',X')$ from $(X',y')$ and then plug $\sigma=\hsigma$ into the construction of confidence sets. As long as $\hsigma$ is consistent, all the asymptotic results in this work still hold.
For high-dimensional linear models, the scaled lasso 
provides a consistent $\hsigma$ \citep{sun12}.

\section{Competing methods}\label{sec:twosteplasso}

To illustrate the effectiveness of our two-step Stein method, we first present three alternative procedures that can be derived by extending ideas from construction of nonparametric regression confidence sets in conjunction with lasso estimation. Since all of them make use of oracle properties, we review an error bound for lasso prediction due to \cite{Bickel09}.

\subsection{Lasso prediction error}
Given $X$, $y$ and $\lambda>0$, consider the lasso estimator $\hbeta=\hbeta(y,X; \lambda)$ defined as in \eqref{eq:lassodef}.
Let $\omega(X) = \max_j (\|X_j\|^2/n)$.
Error bounds of lasso prediction have been established under the restricted eigenvalue assumption \citep{Bickel09}.
For $S\subseteq [p]$ and $c_0>0$, define the cone
\begin{equation}\label{eq:conedef}
\scrC(S,c_0) \defi \left\{\delta\in\R^p:
\sum_{j\in S^c} |\delta_{j}| \leq c_0 \sum_{j\in S}  | \delta_{j}|\right\}.
\end{equation}
We say the design matrix $X$ satisfies $\textup{RE}(s,c_0)$,
for $s \in [p]$ and $c_0>0$, if
\begin{equation}\label{eq:REdef}
\kappa(s,c_0;X) \defi \min_{|S| \leq s} \min_{\delta \ne 0}
\left\{\frac{\|X \delta\|}{\sqn \|\delta_{S} \|}: \delta \in \scrC(S,c_0) \right\} >0.
\end{equation}

\begin{lemma}[Theorem 7.2 in \cite{Bickel09}]\label{lemma:LassoLoss}
Let $n\geq 1$ and $p\geq 2$.
Suppose that $\|\beta\|_0 \leq s$ and $X$ satisfies
Assumption $\textup{RE}(s,3)$. Choose $\lambda=K \sigma \sqrt{\log(p)/n}$
for $K > 2\sqrt{2}$. Then we have
\begin{align}\label{eq:LassoPrediction}
\Prob\left\{\|X(\hbeta-\beta)\|^2
\leq\frac{16K^2\sigma^2 \omega(X)}{\kappa^2(s,3;X)}s\log p\right\}
\geq 1-p^{1-K^2/8}.
\end{align}
\end{lemma}

\begin{remark}\label{rm:sparityrange}
The original theorem in \cite{Bickel09} assumes that all the diagonal elements of the Gram matrix $X^\trans X/n$ are $1$ for simplicity, while we remove this assumption by including the term $\omega(X)$.
\end{remark}

\subsection{Another adaptive method}\label{sec:Robinmethod}

Here we develop another adaptive method following the procedure in Section~3 of \cite{Robins06}, which constructs a confidence set for $\mu$ from $y\sim\dnorm_n(\mu,\sigma^2\bfI_n)$ via sample splitting.
Applied to the linear model~\eqref{highdimensionRL},  the method can be described as follows.
Split the original data set into $(X', y')$ and $(X, y)$, of which the former is used to obtain an initial lasso estimate $\hbeta=\hbeta(y',X';\lambda)$ \eqref{eq:lassodef}, and the latter is used to compute two quantities
\begin{align}\label{eq:lossEst}
R_{n} = \frac{1}{n}\| y - X\hat{\beta} \|^2 - \sigma^2, \quad \quad \hat{\tau}^2_n = \frac{2\sigma^4}{n} + \frac{4\sigma^2}{n^2}\| X\beta - X\hbeta \|^2,
\end{align}
where $R_n$ is an estimate of the loss $\|X\beta - X\hbeta\|^2/n$.
Then, a confidence ball for $\mu=X\beta$ is constructed in the form of
\begin{align}\label{eq:confidencesetrobins}
\wh C_a=\left\{ \mu \in \R^n : \frac{R_{n} -  n^{-1}\| \mu - X\hbeta \| ^2}{\hat{\tau}_n} \geq - z_\alpha \right\},
\end{align}
where $z_\alpha$ is the $(1 - \alpha)$ quantile of the standard normal distribution.
Note that $\hat{\tau}_n$ in \eqref{eq:confidencesetrobins} contains the term $\| \mu - X\hbeta \|$ as well so an explicit form of the confidence ball is
\begin{align*}
 &\left\{  \mu \in \R^n : \frac{1}{n}\| \mu - X\hbeta \|^2  \leq  r^2_a = R_n + O\left( \sqrt{(R_n + 1) / n}\right) \right\},
\end{align*}
where $r_a$ is the radius. 

To establish the convergence rate of the diameter of $\wh C_a$, we need an assumption, similar to RE$(s, c_0)$, on the restricted maximum eigenvalue of $X^\trans X/n$ over the cone $\scrC(S,c_0)$ \eqref{eq:conedef}.
For $s\in[p]$ and $c_0>0$, let
\begin{align*}
\zeta(s, c_0; X) \defi \max_{|S| \leq s} \max_{\delta \ne 0}
\left\{\frac{\|X \delta\|}{\sqn \|\delta_{S} \|}: \delta \in \scrC(S,c_0) \right\}.
\end{align*}

\begin{theorem}\label{thm:CRobins}
The $(1-\alp)$ confidence set $\wh C_a$~\eqref{eq:confidencesetrobins} is honest for all $\beta\in\R^p$.
Suppose $s\log p=o(n)$, the sequence $X=X(n)$ satisfies
\begin{align*}
\liminf_{n \to \infty}\kappa(2s,3;X) = \kappa > 0, 
\quad \limsup_{n \to \infty} \zeta(s, 3; X) = \zeta < \infty,
\quad \limsup_{n \to \infty}\omega(X)= \omega < \infty,
\end{align*}
and so does the sequence $X'=X'(n)$.
Then with a proper choice of $\lmd \asymp\sqrt{\log p/n}$, for any $\beta\in\scrB(s)$ the diameter
\begin{align}\label{eq:adaptive_diameter}
|\wh C_a| = O_p\left(n^{-1/4}+\sqrt{s \log p/n}\right).
\end{align}
\end{theorem}

These properties have been informally discussed in the introduction (Section~\ref{sec:intro}).
Although $\wh C_a$ is also honest over the entire parameter space, the upper bound on its diameter critically depends on the sparsity of $\beta$.
The scaling $s\log p=o(n)$ is the minimum requirement for the lasso to be consistent in estimating $\mu$ or $\beta$.
In general, this scaling is also needed for the RE assumption to hold with $\liminf_{n }\kappa(2s,3;X)  > 0$ \citep{negahban12} and for the upper bound on $|\wh C_a|$ to be informative.
This is different from the universal bound~\eqref{eq:univbound} on $\E|\wh C|^2$ for the two-step method.
The diameter $|\wh C_a|$ adapts to the optimal rate for sufficiently sparse $\beta$ as $s\log p=O(\sqn)$;
see Remark~\ref{rmk:diameterStein} for related discussion. Our numerical results in Section~\ref{subsec:dense} demonstrate that $|\wh C_a|$ can be 10 times larger than the diameter of our two-step Stein method when $\beta$ is not sparse.

\subsection{An oracle lasso method}

We calculate the lasso $\hbeta=\hbeta(y, X; \lambda)$ from the whole data set without sample splitting, which we denote by $(X,y)$ in this subsection.

Assuming the true sparsity $s_{\beta}=\|\beta\|_0$ is known (the oracle), a $(1-\alpha)$ confidence ball for $X\beta$ is constructed as
\begin{align*}
\left\{ \mu \in \R^n : \frac{1}{n}\| \mu - X\hbeta \|^2 \leq c_o(\alpha) \sigma^2\frac{s_{\beta} \log p}{n} \mathop{:=} r^2_o \right\},
\end{align*}
where $c_o(\alpha)$ is a constant depending on the design matrix $X$ and the tuning parameter $\lambda$.
We estimate $c_o(\alpha)$ by a similar procedure to be described in Section~\ref{sec:lasso_algorithm} for a two-step lasso method. 
Although there are sharper upper bounds, e.g. $O(s_{\beta} \log (p/s_{\beta})/n)$, for lasso prediction error (e.g. Chapter~11 in \cite{tibshirani15}), our choice of $\lambda$ is tuned to achieve the desired coverage rate in our numerical results and thus the corresponding $r_o$ is already optimized in this sense.

It should be pointed out that the oracle lasso is {\em not} implementable in practice since the true sparsity $s_{\beta}$ is unknown.
In theory, it can build a confidence set with a diameter on the order of $(s_{\beta}\log p/n)^{1/2}$, potentially faster than the rate $n^{-1/4}$, however, the constant $c_o(\alp)$ can be large and difficult to approximate. Indeed, in comparison with the oracle lasso, our method often constructs confidence sets with a smaller volume even under highly sparse settings, which highlights the practical usefulness of our two-step method.

\subsection{A two-step lasso method}\label{sec:lasso_algorithm}

To appreciate the advantage of using Stein estimates in the shrinkage step of our construction, we compare our method with a two-step lasso method, in which we replace the Stein estimate by the lasso to build a confidence set for $\mu_\perp$, the mean for weak signals.
Consider the two-step method in Section~\ref{sec:singlecandidate} with a given candidate set $A$. Let $k = \rank(X_A)$ and further assume $A$ contains strong signals only, that is, $A \subseteq \supp(\beta)$.
We use the same method to find $\hmu_A$ and $r_A$ \eqref{eq:ra} in the projection step.
Like the oracle lasso, we assume the {true} sparsity $s_{\beta}=\|\beta\|_0$ is given and construct a confidence set for $\mu_\perp$ based on the error bound for lasso prediction.

Apply lasso on $(P_A^\perp X, y_\perp)=(P_A^\perp X, P_A^\perp y)$ with a tuning parameter
\begin{align} \label{eq:lambda2}
\lambda_2 = K\sigma \sqrt{\log (p - k) / n}, \quad  \quad K > 2\sqrt{2},
\end{align}
to find the estimate
\begin{align}\label{weaksignalLasso}
\tilde{\beta} = \tilde{\beta}(\lambda_2)=\argmin_{\beta\in\R^p} \left[\frac{1}{2n} \| y_\perp - P_A^\perp X \beta\|^2 +\lmd_2 \|\beta\|_1\right].
\end{align}
It is natural to estimate the center $\mu_\perp=P_A^\perp\mu$ by the lasso prediction $\hmu_\perp = P_A^\perp X\tilde{\beta}$.
As a corollary of Lemma~\ref{lemma:LassoLoss}, we find an error bound for $\|\hmu_\perp-\mu_\perp\|^2$:

\begin{corollary}\label{thm:hoestCIlasso}
Let $n\geq 1$ and $p\geq 2$. Suppose that $\|\beta\|_0\leq s$ and Assumption $\textup{RE}(s,3)$ holds for $X$. Choose $\lambda_2$ as in \eqref{eq:lambda2}. Then for any fixed $A \subseteq \supp(\beta)$ with $k=\rank(X_A)<s$, we have
\begin{align}\label{eq:Lassoperp}
\Prob\left\{\|P_A^\perp X(\tdbeta-\beta)\|^2
\leq\frac{16K^2\sigma^2 \omega(X)}{\kappa^2(s,3;X)}(s-k)\log (p-k)\right\}
\geq 1-(p-k)^{1-K^2/8}.
\end{align}
\end{corollary}

\noindent
Accordingly, the radius for weak signals is chosen as
\begin{align}
r_\perp^2 = c_2 \tilde{r}_\perp^2 = c_2 c_l(\alpha/2)\sigma^2\frac{(s_{\beta} - k) \log (p - k) }{n}, \label{eq:lassorperp}
\end{align}
where $c_l(\alpha/2) = c_l(\alpha/2; P_A^\perp X)$ is a constant.
Lastly, we combine $(\hmu_\perp, r_\perp)$ with $(\hmu_A, r_A)$ as in \eqref{eq:ellipsoidmu} to define the confidence set $\wh C$.

Again we use sample splitting to define the candidate set ${A}$ by thresholding the lasso estimate 
$\hbeta(y', X'; \lambda)$ in \eqref{eq:lassodef} with a threshold value $\tau = \Omega_p( \|\hbeta - \beta\|_\infty)$ so that $\Prob \left({A} \subseteq \supp(\beta)\right) \to 1$, satisfying the assumption in Corollary~\ref{thm:hoestCIlasso}.
Upper bounds on $\|\hbeta - \beta\|_\infty$ are available under certain conditions; see, for example, Theorem~11.3 in \cite{tibshirani15}.

\begin{remark}\label{rm:lassorates}
Suppose $\beta$ is sufficiently sparse so that $s_{\beta}\log p\ll \sqn$.
Then, it follows that both $r_A$ and $r_\perp$ of the two-step lasso converge faster than the rate of $n^{-1/4}$.
This is not surprising and shows the advantage of the oracle knowledge of the true sparsity $s_{\beta}$. Of course, in practice we do not know $s_{\beta}$ and therefore, this two-step lasso method, like the oracle lasso, is not implementable for real problems.
The numerical comparisons in the next section will show that our two-step Stein method, which does not use the true sparsity in its construction, is more appealing than the two-step lasso:
Its adaptation to the underlying sparsity is comparable to the two-step lasso, while its coverage turns out to be much more robust.
\end{remark}


We follow the same procedure as the two-step Stein method to implement the two-step lasso method with multiple candidate sets $A_m, m=1,\ldots,M$ --- threshold $\hbeta(y', X'; \lambda)$ with a sequence of threshold values to construct $A_m$ \eqref{eq:threshold} and then choose the confidence set with the minimum volume or diameter.
The main difference is how to approximate $c_l(\alpha)$ in \eqref{eq:lassorperp}, which is done by the following approach.

We first use
$b = \max_{i \in [p]} (X_i'^\trans y')/\|X_i'\|^2$
as a rough upper bound for $\|\beta\|_{\infty}$.
For $j = 1,2,\ldots, N$, we draw an $s_{\beta}$-sparse vector, $\gamma_j \in \R^p$, of which the nonzero components follow $\mathcal{U}(-b, b)$.
Then we sample $Y_j^* \sim \dnorm_{n}(X\gamma_j, \sigma^2 \bfI_{n})$ and calculate lasso estimate $\hat{\gamma}_j(\lambda)=\hbeta(Y_j^*,X;\lambda)$ as in \eqref{eq:lassodef} with the tuning parameter $\lambda$ for all $j$.
Let $c_j = {\| X (\hat{\gamma}_j(\lambda) - \gamma_j) \|^2}/{(\sigma^2 s_{\beta} \log p)}$.
For a large $N$, $c_l(\alpha)$ can be approximated by the $(1 - \alpha)$ quantile of $\{ c_j\}$.
Here, $\lambda = \nu \cdot K\sigma^2\sqrt{\log p / n}$, where $\nu \leq 1$ is a pre-determined constant.
This choice is slightly smaller than the theoretical value in Lemma~\ref{lemma:LassoLoss},
but gives a stable estimate of $c_l(\alpha)$ with the desired coverage. 
As we calculate $b$ with $(X',y')$ in the above, our estimate
of $c_l(\alpha)$ is independent of the response $y$.
It is possible that a candidate set $A_m$ defined by \eqref{eq:threshold} may contain $s$ or more predictors.
In this case, we will only include the largest $s-1$ predictors in terms of their absolute lasso coefficients, as Corollary~\ref{thm:hoestCIlasso} requires $|A_m|<s$.

\section{Numerical results}\label{sec:simulation}

We will first compare our method with the above competing methods when $\beta$ is sparse relative to the sample size, i.e., $s/n$ is small, and then consider the more challenging settings in which the sparsity $s$ is comparable to $n$.

\subsection{Simulation setup}
The rows of $X$ and $X'$, both of size $n \times p$, are independently drawn from $\dnorm_p(0, \Sigma)$ and the columns are normalized to have an identical $\ell_2$ norm. 
We use three designs for $\Sigma$ as in \cite{Dezeure15}: 
\begin{align*}
\text{Toeplitz:}& \qquad \Sigma_{i,j} = 0.5^{|i-j|}, \\
\text{Exp.decay:}& \qquad (\Sigma^{-1})_{i,j} = 0.4^{|i-j|}, \\
\text{Equi.corr:}& \qquad \Sigma_{i,j} = 0.8 \text{ for all } i\neq j, \Sigma_{i,i} = 1 \text{ for all } i.
\end{align*}
The support of $\beta$ is randomly chosen and its $s$ nonzero components are generated in two ways:
\begin{enumerate}
  \item They are drawn independently from a uniform distribution $\mathcal{U}(-b, b)$. 
  \item Half of the nonzero components follow $\mathcal{U}(-b, b)$ while the other half following $\mathcal{U}(-0.2, 0.2)$, so there are two signal strengths under this setting.
\end{enumerate}
Lastly, $y$ and $y'$ are drawn from $\dnorm_n(X\beta, \sigma^2 \bfI_n)$ and $\dnorm_n(X'\beta, \sigma^2 \bfI_n)$, respectively.
In our results, we chose $n = n' = 200$, $p=800$, $\sigma^2=1$ and $s=10$, and $b$ took 10 values evenly spaced between $(0, 1)$ and $(1, 5)$.
In total, we had $60$ simulation settings, each including one design for $\Sigma$, one way of generating $\beta$, and one value for $b$. 
Under each setting, $100$ data sets were generated independently, so that the total number of data sets  used in this simulation study was 6,000.

The confidence level $1- \alpha$ was set to $0.95$. 
The threshold values $\{a_m\}$ in \eqref{eq:threshold} were evenly spaced from $0$ to $4$ with a step of $0.05$.
All the competing methods use lasso in some of the steps, and the tuning parameter $\lambda$ was chosen by three approaches:
1) the minimum theoretical value in \cite{Bickel09}, $ \lambda_{val} = 2\sqrt{2}\sigma \sqrt{\log p / n}$, 
2) cross validation $\lambda_{cv}$, and 3) one standard error rule $\lambda_{1se}$.
For the one standard error rule, we choose the largest $\lambda$ whose test error in cross validation is within one standard error of the error for $\lambda_{cv}$.
Since it is time-consuming to approximate $c_o(\alpha) = c_o(\alpha; X, \lambda)$ for the oracle lasso when $\lambda$ is chosen by a data-dependent way, we set $c_o(\alpha; X, \lambda_{cv})=\eta_1c_o(\alpha; X, \lambda_{val})$ and $c_o(\alpha; X, \lambda_{1se})=\eta_2c_o(\alpha; X, \lambda_{val})$, where the factors $\eta_k$ were chosen such that the overall coverage rate across data sets simulated with $b > 0.3$ was around the desired level.

Unlike the adaptive method in Section~\ref{sec:Robinmethod} and our two-step methods, the oracle lasso method does not require sample splitting. 
Consequently, a confidence set is constructed based on the whole data set including both $(X, Y)$ and $(X', Y')$ for a fair comparison.
We compare the geometric average radius $\bar{r} = (r_A^{|A|} r_\perp^{n - |A|})^{1/n}$ of our two-step methods with $r_a$ of the adaptive method and $r_o$ of the oracle lasso.
This is equivalent to comparing the volumes of the confidence sets.

\subsection{Results on the two-step Stein method}\label{sec:simulation_stein}

In this subsection we compare the two-step Stein method with the adaptive method and the oracle lasso. 
The constants $c_1$ and $c_2$ of our method were chosen by minimizing the volume in \eqref{eq:logVolOptimization} with  upper bound $E=10$.

Figure~\ref{fig:radius_stein_1} compares the geometric average radius $\bar{r}$ among the three methods against the signal strength $b$ under the first way of drawing $\beta$. Every point in a panel was computed by averaging $\bar{r}$ from $100$ data sets under a particular simulation setting.
It is seen from the figure that $\bar{r}$ by our method was dramatically smaller than the other two methods for almost every setting.
This suggests that the volumes of our confidence sets were orders of magnitude smaller than the other two methods, as the ratio of the radii will be raised to the power of $n=200$ for comparing volumes.
When $X$ was drawn from the equal correlation (Equi.corr) design, $\bar{r}$ of the oracle lasso and the adaptive methods kept increasing as $b$ increased, while $\bar{r}$ by our method became stable after $b>2$.
Overall, the equal correlation design was more challenging than the other two designs, for which our method outperformed the other two methods with the largest margin.
Unlike the other two methods, our method was less sensitive to the choices of $\lambda$ and the designs of $X$.
Essentially, $r_A$ and $r_\perp$ by our method are determined by the candidate set $A$. 
Even if a different $\lambda$ is used, our method can choose adaptively an optimal $A$ close to $\supp(\beta)$, showing the advantage of using multiple candidate sets.

\begin{figure}[t]
\begin{center}
\includegraphics[width=0.8\textwidth]{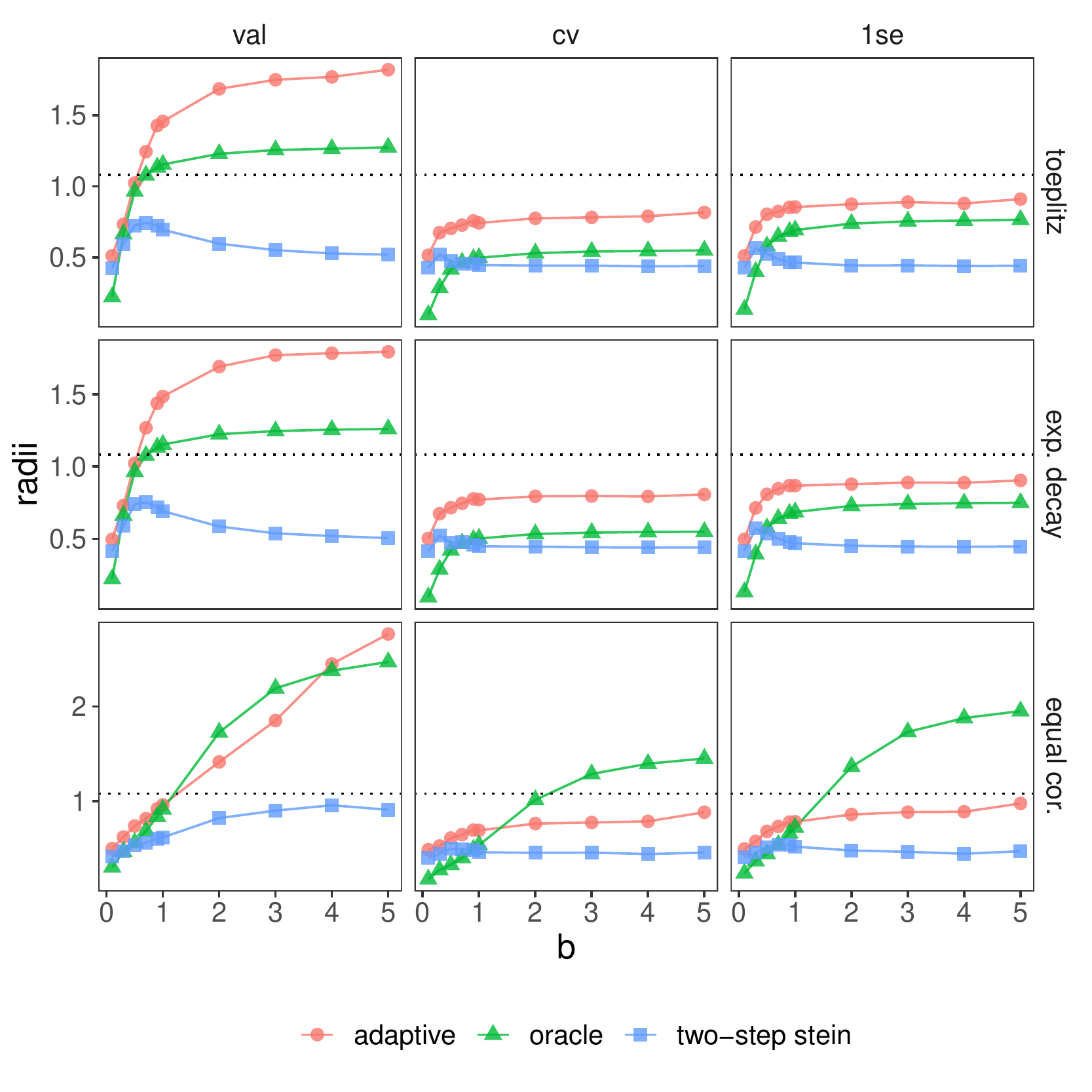}
\caption{Geometric average radius against $b$ under the first way of generating $\beta$. Each panel reports the results for one type of design (row) and one way of choosing $\lambda$ (column), where the dashed line indicates the naive $\chi^2$ radius.}
\label{fig:radius_stein_1}
\end{center}
\end{figure}

In a similar way, Figure~\ref{fig:radius_stein_2} plots $\bar{r}$ against $b$ in the second scenario of drawing $\beta$.
When $b$ is large (e.g, $b \geq 1$), the $\beta$ contains a mixture of weak and strong signals.
Again, we see that $\bar{r}$ of our method was smaller than the other two competitors for most settings. 
The average radius by our method often decreased as $b > 1$, which shows that our method can properly distinguish strong signals and weak signals.

\begin{figure}[htp]
\begin{center}
\includegraphics[width=0.8\textwidth]{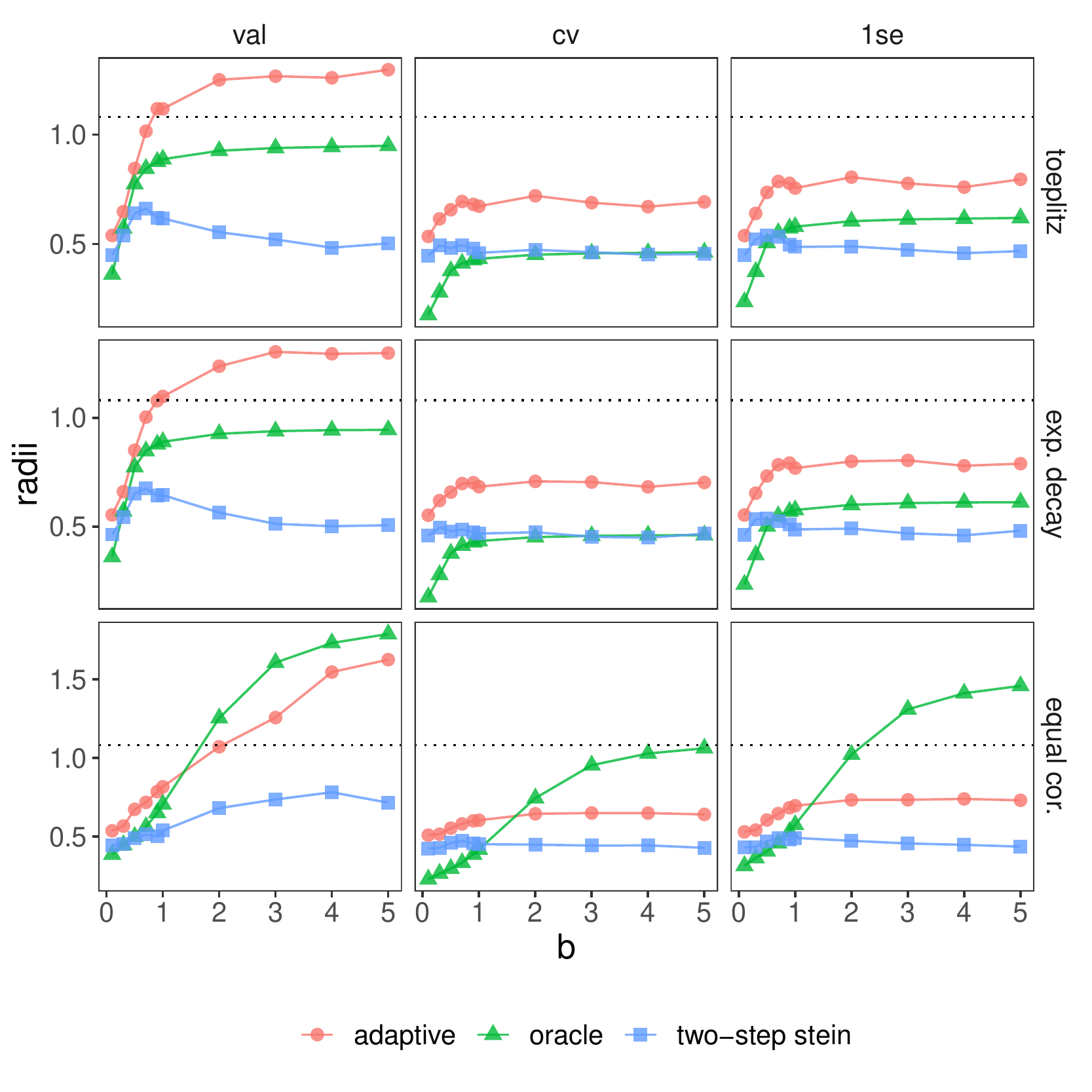}
\caption{Average radius $\bar{r}$ against $b$ in the second scenario of generating $\beta$.}
\label{fig:radius_stein_2}
\end{center}
\end{figure}

The coverage rates, each computed from $100$ data sets, for each of the three ways of choosing $\lambda$ are summarized in Figure~\ref{fig:coverage_stein}. We pooled the results from three types of design matrices together in the figure, because the coverage rates distributed similarly across them.
The coverage rates of our method matched the desired $95\%$ confidence level very well,
with coverage rate $>0.9$ for $96\%$ of the cases.
This result is particularly satisfactory for a quite small sample size of $n=200$. The adaptive method also showed
a good coverage, but slightly more conservative than the desired level. The oracle lasso had the most variable coverage rate across different settings when $\lambda$ was selected in a data-dependent way ($\lambda_{cv}$ or $\lambda_{1se}$). In fact, its coverage could drop below $0.5$ for these two cases (not shown in the figure). This shows the difficulty in practice to construct stable confidence sets using error bounds like \eqref{eq:LassoPrediction} even with a known sparsity.
Together with the results in Figures~\ref{fig:radius_stein_1} and \ref{fig:radius_stein_2}, this comparison demonstrates the advantage of the proposed two-step Stein method: 
It builds much smaller confidence sets, while closely matching the desired confidence level.
In particular, our confidence sets were uniformly smaller than those by the adaptive method (Section~\ref{sec:Robinmethod}) for all simulation settings and all choices of $\lambda$.

\begin{figure}[t]
\begin{center}
\includegraphics[width=0.6\textwidth]{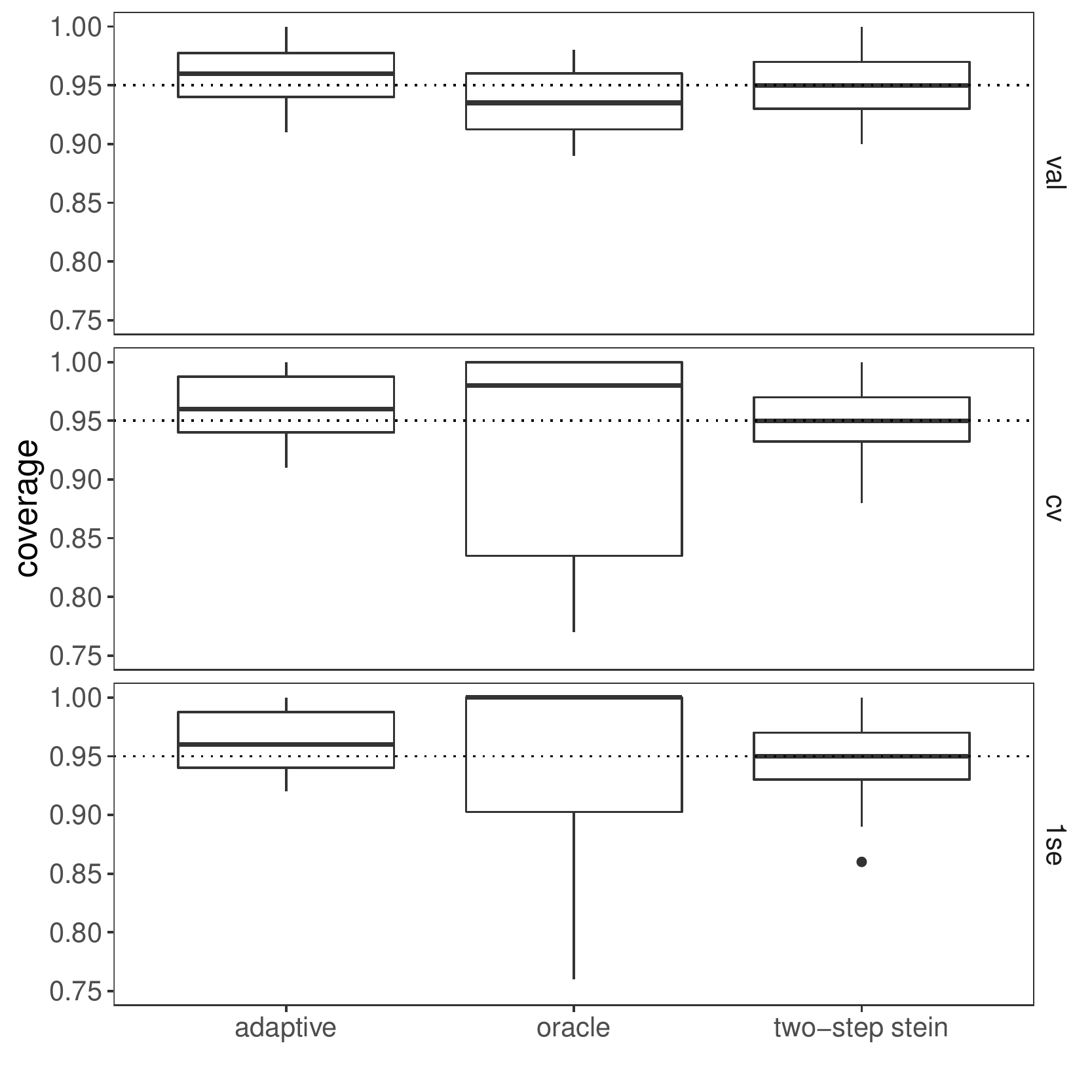}
\caption{Box plots of coverage rates for each choice of $\lambda$, pooling data from three designs.
The dashed lines indicate the desired confidence level of $95\%$.}
\label{fig:coverage_stein}
\end{center}
\end{figure}

\subsection{Comparison with the two-step lasso method}

We discussed in Section~\ref{sec:algorithmstein} two ways to choose $c_1$ and $c_2$, that is, by minimizing the volume or by minimizing the diameter of the confidence set for our proposed two-step framework.
Here we compare the two-step Stein method and the two-step lasso, each with the two ways to choose the constants. 
The two-step Stein method by minimizing the volume (abbreviated as TSV)  is the same method used in the previous comparison. 
Similarly, we use the short-hand TSD, TLV, and TLD for the two-step Stein method by minimizing diameter, the two-step lasso method by minimizing volume and by minimizing diameter, respectively.
The true sparsity $s=10$ was given to the two-step lasso methods.
Only the first scenario of generating $\beta$ was considered in this comparison, since most results in the second scenario were similar.
Figure~\ref{fig:radius_lasso} shows the plots of radius against $b$ by the four methods under different settings, while Figure~\ref{fig:coverage_lasso2} reports the distribution of the coverage rates.
The two-step lasso methods apply the lasso twice, one to generate candidate sets $A_m$ and the other to compute $\hmu_\perp$ and $r_\perp$ for weak signals. 
To clarify, the three ways of choosing $\lambda$ in these figures refer to the step to generate candidate sets $A_m$, while $\lambda_2$ in \eqref{weaksignalLasso} was set to $\nu K \sigma^2 \sqrt{\log (p - |A|) / (n - |A|)}$, where $\nu  = 0.5$ in our simulation.

\begin{figure}[t]
\begin{center}
\includegraphics[width=0.8\textwidth]{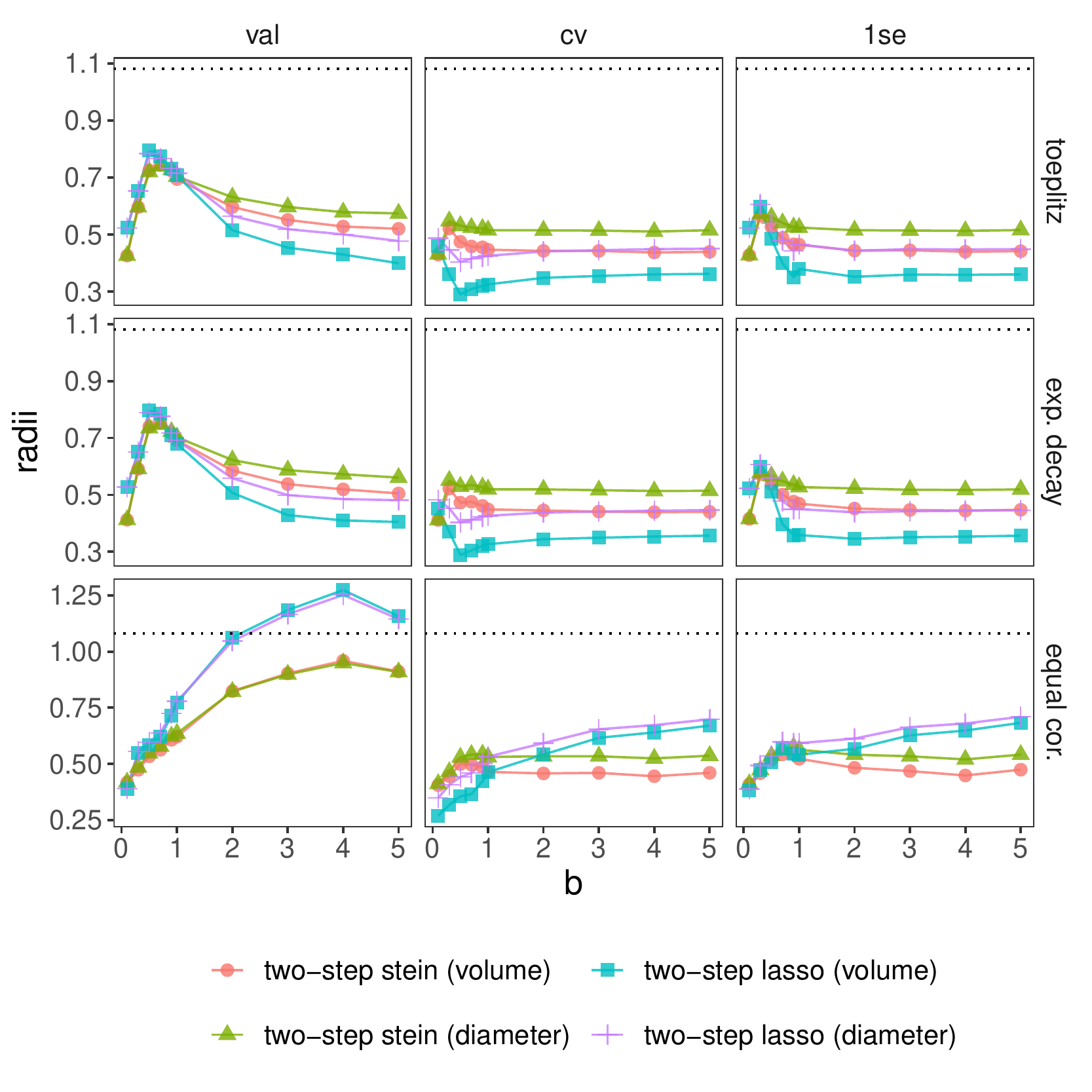}
\caption{Average radius $\bar{r}$ against $b$ in the first scenario of generating $\beta$.} 
\label{fig:radius_lasso}
\end{center}
\end{figure}

\begin{figure}[t]
\begin{center}
\includegraphics[width=0.6\textwidth]{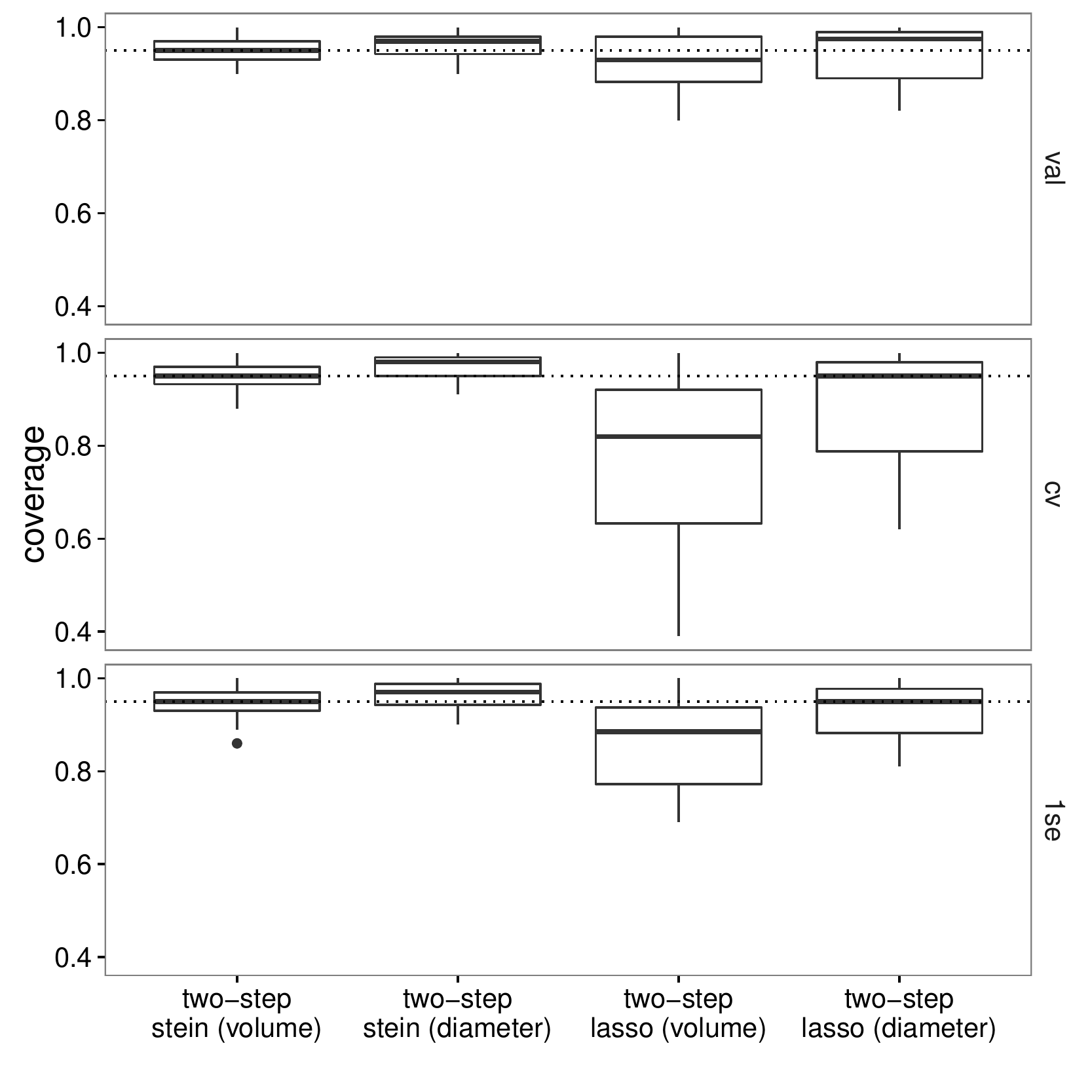}
\caption{Box plots of coverage rates for each choice of $\lambda$. The dashed lines indicate the desired confidence level of $95\%$.}
\label{fig:coverage_lasso2}
\end{center}
\end{figure}

We make the following observations from the two figures.
First, the two-step Stein methods showed a substantially more satisfactory coverage than the two-step lasso methods. 
The coverage was close to 0.95 for both TSV and TSD, while the coverage rates of TLV and TLD had a much larger variance and were especially poor when $\lambda$ was chosen via cross validation. 
The confidence sets by the two-step lasso methods had a slightly smaller average radius than the two-step Stein methods for the Toeplitz and the exponential decay designs. However, given their low and unstable coverage rates, this does not imply the two-step lasso methods constructed better confidence sets.
Recall that $|\cset| = O_p(n^{-1/4}\vee \sqrt{s/n})$ for the two-step Stein methods and $|\cset| = O_p(\sqrt{s\log p /n})$ for the two-step lasso methods. 
The signals were very sparse in our simulation, with $s=10$ much smaller than $p$, favorable for the two-step lasso methods.
Even so, we find the two-step Stein methods very competitive, noting that the radii of both TSV and TSD were actually comparable or slightly smaller than the two-step lasso methods for the equal correlation designs, in which the predictors were highly correlated. 
This comparison demonstrates that the two-step Stein method is more appealing in practice, as it does not require any prior knowledge about the underlying sparsity but gives a better and more stable coverage.
Second, both ways of choosing the constants $c_1$ and $c_2$ worked well for the two-step Stein method. 
On the contrary, it is seen from Figure~\ref{fig:coverage_lasso2} that the coverage rate of TLV was significantly lower than that of TLD in the bottom two panels. 
Lastly, between using $\lambda_{cv}$ and $\lambda_{1se}$ in the lasso for defining candidate sets $A_m$, we recommend the latter, as it tends to give comparable radii but a better coverage, especially for the two-step lasso.

We also compared the performance between the oracle lasso method and TLD, both constructing confidence sets based on the lasso prediction \eqref{eq:LassoPrediction} with a known sparsity. 
The coverage rates of the two methods were quite comparable as reported in Figures~\ref{fig:coverage_stein} and \ref{fig:coverage_lasso2}. 
The geometric average radius of the oracle lasso method (Figure~\ref{fig:radius_stein_1}) was $2$ to $5$ times that of TLD (Figure~\ref{fig:radius_lasso}). 
The difference was especially significant when the signal strength was high (large $b$). 
This comparison confirms that, by separating strong and weak signals, our two-step framework can greatly improve the efficiency of the constructed confidence sets.

\subsection{Dense signal settings} \label{subsec:dense}

We have shown the advantages of our two-step Stein method in the last two subsections under sparse settings. 
Recall that the dimension of our data was $(n,p)=(200,800)$ with sparsity $s=10$ for $\beta$ in the previous comparisons. 
The goal of this subsection is to illustrate the stable performance of our method when the true signal is dense. 
As such, we changed the sparsity to $s=100$ for the first way of generating $\beta$ and $s=200$ for the second way of generating $\beta$. 
We focused on the equal correlation design, which was the most difficult one among the three designs. 
With the same set of values for the signal strength $b$, we had 20 distinct parameter settings for data generation in this comparison, and again we simulated 100 data sets under each setting. 
The tuning parameter $\lambda$ was selected as $\lambda_{1se}$ for all the results here. 

\begin{figure}
\begin{center}
\includegraphics[width=0.7\textwidth]{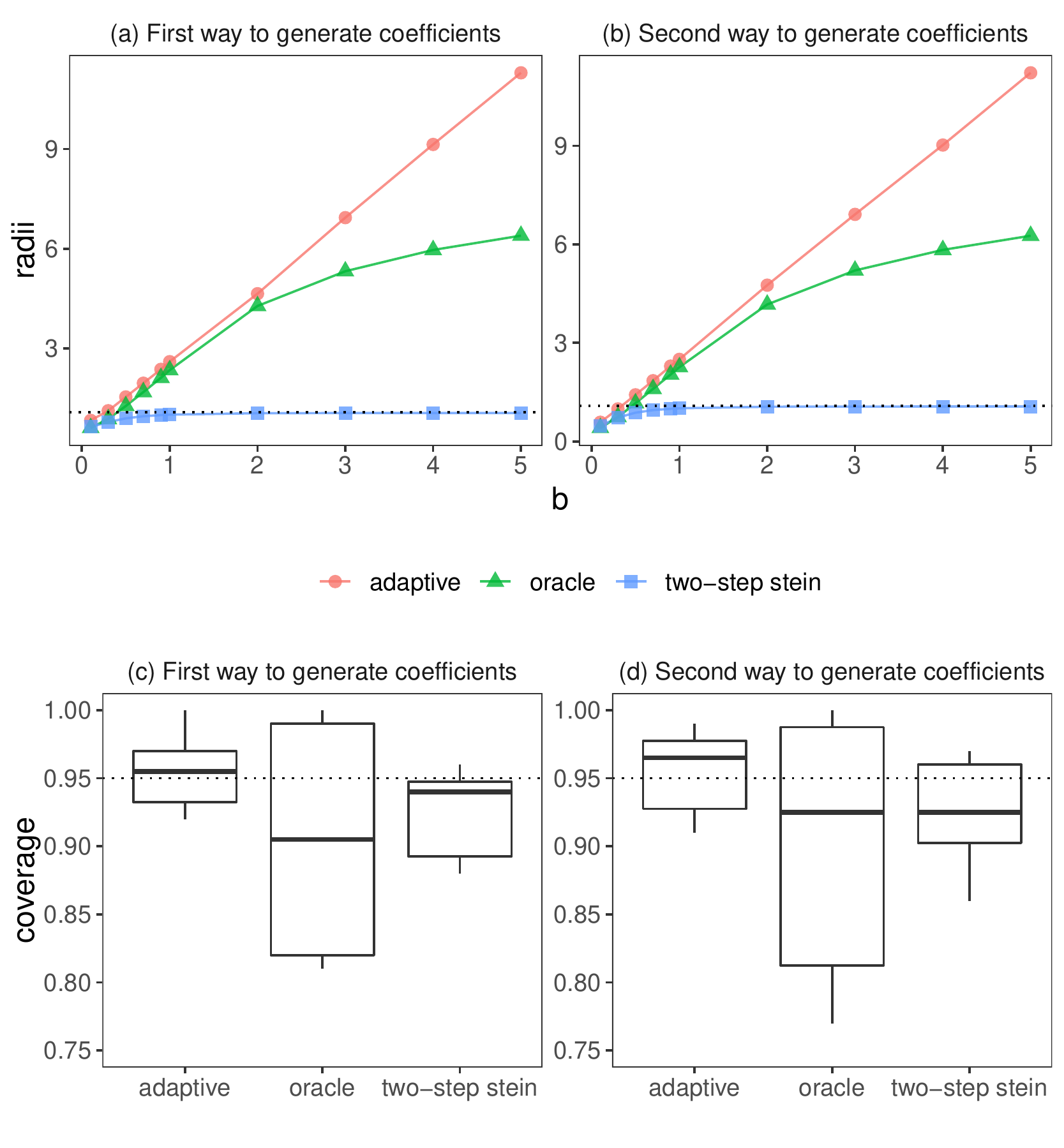}
\caption{Comparison results under dense signal settings. (a) and (b) Geometric average radius against $b$.
(c) and (d) Box plots of the coverage rates. }
\label{fig:dense}
\end{center}
\end{figure}

Figure~\ref{fig:dense} compares the geometric average $\bar{r}$ against $b$ and the coverage among the adaptive method, the oracle lasso and our two-step Stein method. 
In all the scenarios reported in panels~(a) and (b), our method outperformed the other two methods with very big margins in terms of the volume of a confidence set.
For $b>1$, the radius of our method approached the naive radius ${(\chi^2_{n,\alp}/n)^{1/2}}$ as suggested by Theorem~\ref{thm:twostephonest}, while the radii of the oracle lasso and the adaptive methods kept increasing to much greater than the naive $\chi^2$ radius.
This shows that the two competing methods failed to construct acceptable confidence sets when the signal was dense. Since the sparsity level $s$ is comparable to $n$ for the data sets here, the upper bounds for the diameters of these two methods, $|\wh C_o|=O_p(\sqrt{s\log p /n})$ and $|\wh C_a|=O_p(n^{-1/4} + \sqrt{s\log p /n})$, are no longer useful or even valid.
It is seen from Figure~\ref{fig:dense}(c) and (d) that the coverage rates of the two-step Stein method were much better than the oracle lasso, but slightly lower than the adaptive method. 
Nevertheless, our confidence sets still maintained a minimum coverage of $0.9$ in most cases, which is quite satisfactory given the way smaller diameters than the adaptive method. 

\begin{figure}%
    \centering
    \includegraphics[width=0.6\textwidth]{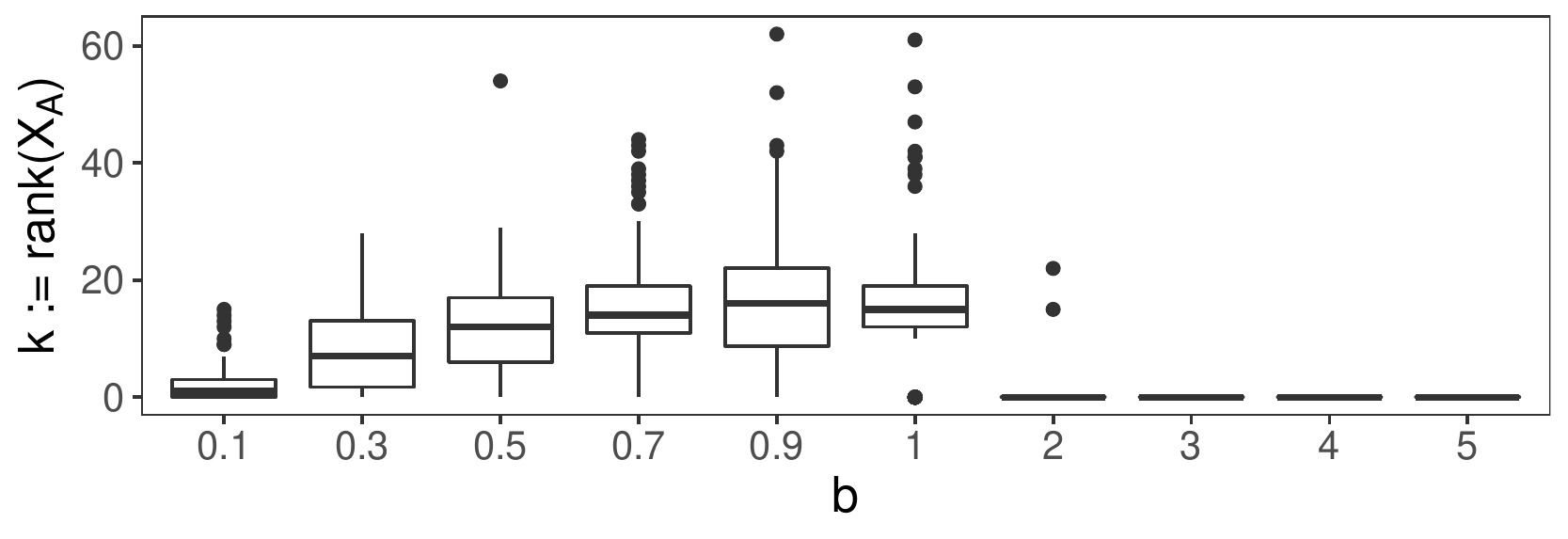}
    \caption{The box plot of $k$ across data sets for each value of $b$}%
    \label{fig:k_b}%
\end{figure}

To understand the behavior of our method in this dense signal setting, we examined the number of variables selected as strong signals in the set $A$, i.e., $k=|A|$. Figure~\ref{fig:k_b} displays the box plot of $k$ across 100 data sets for each value of $b$ under the first way to generate $\beta$. When $b\leq 1$, our two-step method still chose a nonempty candidate set, but $k$ dropped to $0$ for $b \geq 2$, i.e., $A=\varnothing$. Note that the radius of our method will be close to the naive $\chi^2$ radius when $k=n$ or $k=0$; 
see \eqref{eq:expection_diameter} in Theorem~\ref{thm:twostephonest}. 
When the signal strength $b\leq 1$, some small nonzero coefficients are close to zero so $\beta$ is effectively quite sparse, in which case the lasso can select a good subset $A$ of strong signals. 
On the contrary, when $b$ is large, the lasso will not be able to select a majority of the strong signals, leaving $\|\mu_\perp\|=\|P_A^\perp \mu\|$ too big. 
In this setting, our method automatically adjusts its ``optimal" choice to $A=\varnothing$, constructing a confidence set centered at the Stein estimate $\hmu(y;0)$ \eqref{steinestimate} with radius estimated via the SURE.



\section{Discussion}\label{sec:discussion}

For high-dimensional regression, oracle inequalities for sparse estimators cannot be directly utilized to construct honest and adaptive confidence sets due to the unknown signal sparsity.
To overcome this difficulty, we have developed a two-step Stein method, via projection and shrinkage, to construct confidence sets for $\mu=X\beta$ in \eqref{highdimensionRL} by separating signals into a strong group and a weak group.
Not only is honesty achieved over the full parameter space $\R^p$, but also our confidence sets can adapt to the sparsity and strength of $\beta$.
We also implemented an adaptive way to choose a proper subspace for the projection step among multiple candidate sets, which protects our method from a poor separation between strong and weak signals.
Our two-step Stein method showed very satisfactory performance in extensive numeric comparisons, outperforming other competing methods under various parameter settings. 

The focus of this work is on the confidence set for $\mu=X\beta$.
Although related, it is different from the problem of inference on $\beta$. 
In general, it is difficult to infer a confidence set for $\beta$ from the confidence set for $X\beta$ without any constraint on $X$ and $\beta$, because $X$ does not have a full column rank under the high-dimensional setting. 
However, if we know that $\|\beta\|_0\leq s$, then a confidence set $\wh C$ for $\mu$ can be converted into a confidence set for $\beta$ as
$\wh B:=\{\beta\in \scrB(s): X\beta \in \wh C\}$,
which is the union of $s$-dimensional subspaces intersecting $\wh C$. 
It is interesting future work to study the convergence rate of $\wh B$ and related computational issues, such as how to draw $\beta$ from $\wh B$.
On the other hand, if $X$ satisfies $\textup{SRC}(s, c_*, c^*)$, then 
\begin{align*}
c^*\| \beta \| ^2 \geq  \| X \beta \|^2/n,\quad \forall\; \beta \in \scrB(s).
\end{align*}
A hypothesis test about the mean $X\beta$ can be carried out by using the confidence set $\wh C$ to obtain a lower bound on  $\| X \beta \|$, which carries over to a lower bound on $\|\beta\|$ with the above inequality and thus can be used to perform a test about $\beta$. 
See \cite{Nickl13} for a related discussion. We have also demonstrated that our method works well even when the underlying $\beta$ is dense, e.g. $\|\beta\|_0\asymp n$, which is important for practical applications.
See \cite{bradic18} for recent theoretical results on high-dimensional inference for non-sparse $\beta$.

Another future direction is to incorporate the confidence set $\wh C$ with the method of estimator augmentation \citep{zhou14, zhou17} for lasso-based inference. 
Estimator augmentation can be used to simulate from the sampling distribution of the lasso without solving the lasso problem repeatedly, provided a point estimate of $\mu=X\beta$. 
Given $\wh C$, one may randomize the point estimate of $\mu $ by sampling from the confidence set, 
which has been shown to improve the inferential performance of estimator augmentation \citep{min09}.

\appendix
\section{Proofs}\label{sec:proof}
\begin{proof}[Proof of Lemma~\ref{lm:radiusrate}]
By the law of large number, we have
\begin{align}\label{eq:chisqorder}
\frac{\chi^2_{k, \alpha} - k}{\sqrt{2k}} = o(1) + \Phi^{-1}(\alpha) 
\Rightarrow \chi^2_{k, \alpha} = k + o(\sqrt{2k}) + \sqrt{2k} \Phi^{-1}(\alpha)  \asymp k,
\end{align}
where $\Phi^{-1}$ is the inverse of the cumulative distribution function of $\dnorm(0,1)$.
It follows from \eqref{eq:ra} and \eqref{eq:chisqorder} that
\begin{align}\label{eq:raratio}
r_A^2 = c_1 \cdot \sigma^2 \chi^2_{k, \alpha} / n  \asymp k / n.
\end{align}
Let $\veps_\perp=P_A^\perp \veps$. Under the normality assumption of $\veps$, we have
\begin{align*}
    1/B &= \frac{\|y_\perp\|^2}{(n-k)\sigma^2} = 
    \frac{\norm{\veps_\perp}^2 + 2\langle \mu_\perp, \veps_\perp \rangle + \norm{\mu_\perp}^2 }{(n - k)\sigma^2}  \\
    &= 1 + O_p\left(\frac{1}{\sqrt{n - k}}\right) + O_p\left( \frac{\norm{\mu_\perp}}{n - k} \right) + \frac{\norm{\mu_\perp}^2}{(n-k)\sigma^2}.
\end{align*}
It follows, by noting $\norm{\mu_\perp} = o(\sqrt{n - k})$, that 
\begin{align} \label{eq:Lorder}
    \hL = 1 - B = O_p\left(\frac{1}{\sqrt{n - k}}\right) + O_p\left( \frac{\norm{\mu_\perp}^2}{n - k} \right).
\end{align}
By plugging \eqref{eq:Lorder} in \eqref{eq:rperp}, we obtain
\begin{align}\label{rperpratio}
r^2_\perp & = c_2 \cdot \sigma^2\frac{n - k}{n} \left\{O_p\left(\frac{1}{\sqrt{n-k}}\right) + O_p\left(\frac{\norm{\mu_\perp}^2}{n - k}\right) + c_s(\alpha/ 2) \frac{1}{\sqrt{n - k}} \right\} \nonumber\\
& = O_p \left(\frac{\sqrt{n - k}}{n}\right) + O_p\left( \frac{\norm{\mu_\perp}^2}{n} \right).
\end{align}
If $k=O_p(\sqn)$ and $\|\mu_\perp\|=O(n^{1/4})$, it follows from \eqref{eq:raratio} and \eqref{rperpratio} that $|\cset|\asymp_p n^{-1/4}$.
\end{proof}

\begin{proof}[Proof of Theorem~\ref{thm:optimalA}]
Under sparse Riesz condition, letting $G= A^c \cap \supp (\beta)$, we have
\begin{align*}
\norm{\mu_\perp} = \norm{P_A^\perp X_{A^c} \beta_{A^c}}  = \norm{P_A^\perp X_{G}\beta_{G}} \leq  c^*\sqrt{n}\norm{\beta_{G}} =  c^*\sqrt{n}\norm{\beta_{A^c}},
\end{align*}
which, together with $k=o(n)$ and $\|\beta_{A^c}\| = o(1)$, implies $\|\mu_\perp\| = o(\sqrt{n}) = o(\sqrt{n-k})$. Thus, by  Lemma~\ref{lm:radiusrate}, $r_\perp^2 =  O_p(n^{-1/2}+\norm{\beta_{A^c}}^2)$ and the rest of the proof is straightforward.
\end{proof}

\begin{proof}[Proof of Corollary~\ref{cor:lassodi}]
Under the choice of $\lambda$ in this corollary and the assumption that $s \leq (s^* - 1) / (2 + 4c^*/c_*)$, 
Theorem~1 and Theorem~3 in \cite{Zhang08} imply that, for any $\eps > 0$, there exists $N$ such that when $n > N$,
\begin{align}\label{eq:lassothm}
\Prob\left\{ |A| \leq M^*_1 s \text{ and } \| \hat{\beta}-\beta \| \leq M^*_2\sigma \sqrt{(s\log p)/n} \right\} > 1- \eps,
\end{align}
where $M_1^*$ and $M_2^*$ are two constants depending on $c_0$, $c_*$ and $c^*$.
It follows from \eqref{eq:lassothm} that
\begin{align*}
k \leq |A| = O_p(s) = o_p(n),\quad\| \hat{\beta}-\beta \| = O_p \left(\sqrt{s\log p/n} \right).
\end{align*}
Thus, we have
\begin{align} \label{eq:lasso_error_bound}
 \| \beta_{A^c} \| \leq \| \hat{\beta}-\beta \| = O_p\left(\sqrt{s\log p/n}\right)=o_p(1).
\end{align}
Now, all the conditions in Theorem~\ref{thm:optimalA} are satisfied, leading to \eqref{eq:lassosteinradius}.
Further, \eqref{eq:lasso_error_bound} implies that $S_0\subset A$ and thus  
$\| \beta_{A^c}\| = \| \beta_{A^c \cap A_0}\| \leq \| \beta_{A_0 \backslash S_0} \|= O(n^{-1/4})$ with probability at least $1-\eps$.
Consequently, \eqref{eq:lassosteinradiussep} follows from \eqref{eq:steinsparserate}.
\end{proof}

\begin{proof}[Proof of Corollary~\ref{cor:mcpsmp}]

If $\Prob(A=A_0)\to 1$, then the rate of $|\cset|$ in \eqref{eq:mcprate} follows immediately from \eqref{eq:steinsparserate} in Theorem~\ref{thm:optimalA}. Thus, it remains to show that $\hbeta^{\textup{mcp}}_{\lambda_n,\gamma_n}=\hbeta^{\textup{mcp}}_{\lambda, \gamma} (y',X')$ \eqref{eq:mcpdef} is model selection consistent by verifying the conditions of the following corollary, which is a simplified version of Corollary~4.2 in \cite{huang12}.

\begin{corollary} \label{cor:mcp}
Let $\lambda_{\min}$ be the smallest eigenvalue of $(X'_{A_0})^\trans X'_{A_0} / n$,
$\tau_n = \sigma \sqrt{2\log s / (n\lambda_{\min})}$ and
$\lambda^* = 2\sigma \sqrt{2c^* \log(p-s) / n}$.
Suppose that $X'$ satisfies $\textup{SRC}(s^*, c_*, c^*)$, where $0<c_*<c^*$ are constants and $s^* \geq (c^*/c_* + 1/2)s$. If a sequence of $(\lambda_n, \gamma_n)$ satisfies 
$\inf_{A_0} |\beta_j | \geq \gamma_n \lambda_n + a_n\tau_n$ with $a_n \to \infty$,
$\lambda_n \geq a_n\lambda^*$,
$n\lambda_n^2 / (4c^*) > \sigma^2$ and 
$\gamma_n \geq c_*^{-1}\sqrt{4 + c_* / c^*}$, then $\Prob\{\supp(\hbeta^{\textup{mcp}}_{\lambda_n,\gamma_n})=A_0\}\to 1$.
\end{corollary}

Under the SRC assumption 
$\lambda_{\min}$ is bounded from below by $c_*>0$.
It follows from $\tau_n = O(\sqrt{\log s / n})$, $\lambda^* = O(\sqrt{\log p / n})$ and $\lambda_n \gg \sqrt{\log p / n}$ that there exists $a_n \to \infty$ such that $\lambda_n \geq a_n (\lambda^*\vee\tau_n)$. Then we have the following: $\inf_{A_0} |\beta_j | \geq (\gamma_n + 1) \lambda_n \geq \gamma_n \lambda_n + a_n\tau_n$, $\lambda_n \geq a_n \lambda^*$, and $n \lambda_n^2 / (4c^*) \gg \log p > \sigma^2$ when $n$ is sufficiently large. Thus all the conditions in Corollary~\ref{cor:mcp} are satisfied under the assumptions of Corollary~\ref{cor:mcpsmp}. This completes the proof.

Technically, we did not invoke the assumption $s \log p = o(n)$ in the proof. But it is required for the sparse Riesz condition to hold (e.g. for Gaussian designs).
\end{proof}

\begin{proof}[Proof of Lemma~\ref{lemma:consistent_chi}]
We have the following inequalities for any positive $x$ and degree of freedom of $n$ from Lemma~1 in \cite{laurent00}:
\begin{align}
\Prob \left\{ \chi^2_n - n \geq 2\sqrt{n}\sqrt{x} + 2x \right\} &\leq e^{-x}, \label{eq:chi2_gt}\\
\Prob \left\{  \chi^2_n - n  \leq -2\sqrt{n}\sqrt{x} \right\} &\leq e^{-x}. \label{eq:chi2_lt}
\end{align}
The solutions of
$2\sqrt{n}\sqrt{x_1} + 2x_1 = \sqrt{n}\delta$ and  $2\sqrt{n}\sqrt{x_2} = \sqrt{n}\delta$
are plugged in \eqref{eq:chi2_gt} and \eqref{eq:chi2_lt} to obtain
\begin{align*}
& \Prob \left\{ \frac{\chi^2_n}{n} - 1 \geq \sqrt{n}\delta \right\} \leq \exp\left\{ -\frac{(\sqrt{1+2\delta/\sqrt{n}} - 1)^2}{4}n \right\}, \\
& \Prob \left\{ \frac{\chi^2_n}{n} - 1 \leq -\sqrt{n}\delta \right\} \leq \exp\left\{ -\frac{\delta^2}{4} \right\},
\end{align*}
so that
\begin{align*}
\Prob \left\{ \sqrt{n}\left| 1 - \frac{1}{n} \chi^2_{n} \right| \geq \delta \right\}
\leq
2\exp\left\{ -\frac{(\sqrt{1+2\delta / \sqrt{n}} - 1)^2}{4}n \right\}.
\end{align*}
To finish the proof, we will show that
\begin{align}
f(n) = \left(\sqrt{1+2\delta / \sqrt{n}} - 1 \right)^2 n \label{eq:fn_monotonical}
\end{align}
is bounded by $\delta^2$ for any $n$. Replacing $\sqrt{1 + 2\delta/\sqrt{n}}$ with its Taylor expansion $1 + \delta/\sqrt{n} + O(\delta^2/n)$ in \eqref{eq:fn_monotonical}, we get $f(n) = \delta^2 + O(n^{-1/2})\to \delta^2$, as $n\to\infty$.
If $f(n)$ is monotonically increasing in $n$, then $\delta^2$ is a tight upper bound of $f(n)$ for all $n$. Lastly, to prove the monotonicity, it suffices to show the derivative
\begin{align*}
f'(n) = 2 + \delta/\sqrt{n} - \frac{2 + 3 \delta/\sqrt{n}}{\sqrt{1 + 2 \delta/\sqrt{n}}} \geq 0,
\end{align*}
which can be verified easily. Now the proof is completed.
\end{proof}

\begin{proof}[Proof of Lemma~\ref{thm:consistency}]
Let
\begin{align*}
Q(A) &= \E\| P_A^\perp y \|^2
= \E\| P_A^\perp(\mu + \veps) \|^2 \\
&= \| P_A^\perp\mu \|^2 + {\tr}(P_A^\perp)\sigma^2
= \|P_A^\perp \mu \|^2 + (n - k)\sigma^2.
\end{align*}
A few steps of derivation shows that
\begin{align}
& \sigma^2 \hL-{(n - k)^{-1}}\|\hmu_\perp-\mu_\perp\|^2 \nonumber\\
& \quad = \sigma^2 - \frac{\sigma^4(n - k)}{\|P_A^\perp y\|^2} - \frac{1}{n - k} \left\|\left(1-\frac{(n-k)\sigma^2}{\|P_A^\perp y\|^2}\right) P_A^\perp y - P_A^\perp \mu \right\|^2 \nonumber\\
& \quad = \sigma^2 - \frac{1}{n - k} \| P_A^\perp \veps \|^2 + \frac{2\sigma^2}{\|P_A^\perp y\|^2}
\left( \langle \veps, P_A^\perp \mu \rangle + \norm{P_A^\perp \veps}^2 - \sigma^2(n-k) \right).  \label{eq:sure_expansion}
\end{align}
It follows from \eqref{eq:sure_expansion} that
\begin{align*}
& \Prob \left\{ \sup_{A \in \Hset } \sqrt{n-k} \left|\sigma^2 \hL-{(n - k)^{-1}}\|\hmu_\perp-\mu_\perp\|^2\right| \geq \sigma^2\delta \right\}  \\
& \quad \leq \Prob \left\{ \sup_{A \in \Hset } \sqrt{n-k} \left| \sigma^2 - \frac{1}{n - k} \| P_A^\perp \veps\|^2 \right| \geq \sigma^2\delta/2 \right\}  \\
& \qquad + \Prob \left\{ \sup_{A \in \Hset } \sqrt{n-k} \left| \frac{2\sigma^2}{\|P_A^\perp y\|^2}
\left( \langle \veps, P_A^\perp \mu \rangle + \norm{P_A^\perp \veps}^2  - \sigma^2(n-k) \right) \right| \geq \sigma^2\delta/2 \right\},
\end{align*}
where the second probability on the right hand side is bounded by
\begin{align}
& \sum_{A \in \Hset} \Prob \left\{ \left| \frac{2\sigma^2}{\|P_A^\perp y\|^2}
\left( \langle \veps, P_A^\perp \mu \rangle + \norm{P_A^\perp \veps}^2  - \sigma^2(n-k) \right) \right| \geq \frac{\sigma^2\delta}{2\sqrt{n-k}} \right\} \nonumber\\
&\quad \leq \sum_{A \in \Hset}  \left[
\Prob \left\{ \|P_A^\perp y\|^2 \leq \frac{1}{2}Q(A) \right\}\right. \nonumber \\
& \quad\quad\quad\quad \left.
+ \Prob \left\{ 2 \left| \norm{P_A^\perp \veps}^2  - (n-k)\sigma^2 \right| \geq \frac{\delta Q(A)}{2^3\sqrt{n-k}} \right\} \right. \nonumber \\
&\quad\quad\quad\quad \left. + \Prob \left\{ 2 \left| \langle \veps, P_A^\perp\mu \rangle \right| \geq \frac{\delta Q(A)}{2^3\sqrt{n-k}} \right\}
\right]. \label{eq:prob_split}
\end{align}
To prove the theorem, it suffices to show that all three probabilities in \eqref{eq:prob_split} can be bounded by
either $D/(n-k)^2$ or $D/\delta^4$ for some constant $D>0$. Before that, we introduce the following three inequalities derived from Theorem~2 in \cite{Whittle60}:
\begin{align}
\E \left( \|P_A^\perp y\|^2 - Q(A)\right)^4 &\leq D_1 \left[ \sigma^4(n-k)^2 + \|P_A^\perp \mu\|^4 \right], \label{eq:whittle1}\\
\E \left( \norm{P_A^\perp \veps}^2 - (n-k)\sigma^2\right)^4 &\leq D_1 \sigma^4 (n - k)^2, \label{eq:whittle2}\\
\E \left(\langle \veps, P_A^\perp\mu \rangle\right)^4 &\leq D_1 \|P_A^\perp \mu\|^4, \label{eq:whittle3}
\end{align}
for some constant $D_1$ depending on the moments of $\veps_i$.
In our case, $D_1$ only depends on the upper bound $d$ of the eighth moment.
The first term of \eqref{eq:prob_split} can be bounded by
\begin{align*}
& \Prob \left\{ \|P_A^\perp y\|^2 \leq \frac{1}{2}Q(A) \right\} \leq \Prob\left\{\left| \|P_A^\perp y\|^2  - Q(A) \right| \geq \frac{1}{2}Q(A)\right\} \\
& \quad\leq \frac{\E\left( \|P_A^\perp y\|^2 - Q(A)\right)^4}{\left(\frac{1}{2}Q(A)\right)^4} \qquad \text{by Chebyshev inequality}\\
& \quad\leq 16 D_1 \frac{\sigma^4(n-k)^2 + \| P_A^\perp\mu \|^4}{Q(A)^4} \qquad \text{by \eqref{eq:whittle1}}\\
& \quad\leq \frac{16D_1}{(n - k)^2}.
\end{align*}
Similarly, using \eqref{eq:whittle2} and \eqref{eq:whittle3}, we can also show that both the second and the third terms are bounded by $D_2/(\sigma^2\delta^4)$ for some $D_2>0$ depending only on $d$.
Lastly, the proof is finished by letting $D = (16D_1) \vee (D_2 / \sigma^2)$.
\end{proof}


\begin{proof}[Proof of Theorem~\ref{thm:CRobins}]
The honesty of $\wh C_a$ in \eqref{eq:confidencesetrobins} is guaranteed by Theorem~3.1 and Proposition~2.1 in \cite{Robins06} with the only assumption $y/\sqn\sim \dnorm_n(\mu/\sqn,\sigma^2\bfI_n/n)$.
It is not difficult to verify that $(X', y')$ satisfies all the conditions in Corollary~B.2 and Theorem~7.2 of \cite{Bickel09}.
Thus, with probability approaching one, we have $\|\hbeta - \beta\|^2 = O(s \log p / n)$ and $(\hbeta - \beta) \in \scrC(A_0, 3)$, as defined in {\eqref{eq:conedef}}, with $A_0=\supp(\beta)$. By the definition of $\zeta(s,3;X)$, this implies that
\begin{align}\label{eq:sample_splittng_pred_error}
\frac{1}{n}\|X(\beta - \hbeta) \|^2 \leq \zeta \|\hbeta - \beta\|^2=O_p(s \log p / n)=o_p(1).
\end{align}
Again, by Theorem~3.1 in \cite{Robins06}, we have
\begin{align*}
|\wh C_a|^2 = O_p\left(n^{-1/2} + \frac{1}{n}\|X(\beta - \hbeta) \|^2\right)  = O_p\left(n^{-1/2}+s \log p / n\right),
\end{align*}
which completes the proof.
\end{proof}

\begin{proof}[Proof of Corollary~\ref{thm:hoestCIlasso}]
Rewrite orthogonal matrix $P_A^\perp = V V^\trans$, where $V \in \R^{n \times (n-k)}$ consists of orthogonal unit column vectors.
Write the lasso estimate in \eqref{weaksignalLasso} as $\tdbeta=F(y_\perp,P_A^\perp X; n\lmd_2)$,
where $F$ is understood as a mapping with a parameter $n\lmd_2>0$. Since $P_A^\perp X_A=0$, the loss in \eqref{weaksignalLasso} becomes
\begin{align*}
\frac{1}{2n} \| y_\perp - P_A^\perp X \beta\|^2 +\lmd_2 \|\beta\|_1
&= \frac{1}{2n} \| y_\perp - P_A^\perp X_{A^c} \beta_{A^c}\|^2  +\lmd_2 \|\beta\|_1 \\
&= \frac{1}{2n} \| V^\trans y - V^\trans X_{A^c} \beta_{A^c}\|^2  +\lmd_2 \|\beta\|_1,
\end{align*}
which demonstrates that $\tilde{\beta}_A=0$ and $\tilde{\beta}_{A^c}=F(V^\trans y,V^\trans X_{A^c}; n\lmd_2)$.
Moreover, we have
\begin{align}\label{eq:equiprloss}
\| V^\trans X_{A^c} (\tdbeta_{A^c}-\beta_{A^c})\|=\| P_A^\perp X (\tdbeta-\beta)\|.
\end{align}

We will verify that the lasso problem,
{$\tilde{\beta}_{A^c}=F(V^\trans y,V^\trans X_{A^c};n\lmd_2)$,}
satisfies all the assumptions in Lemma~\ref{lemma:LassoLoss}
so that we can apply \eqref{eq:LassoPrediction} to bound the prediction error on the left side of \eqref{eq:equiprloss}.
Since $A \subseteq \supp(\beta)$, we have $\|\beta_{A^c}\|_0 \leq s - k$. Next, we show $V^\trans X_{A^c} \in \R^{(n-k)\times (p-k)}$ satisfies $\textup{RE}(s-k, 3)$.
Let $D$ be any subset of $[p-k]$ such that $|D| \leq (s - k)$.
For any nonzero $\gamma \in \R^{p-k}$ in the cone $\scrC(D, 3)$, a vector $\delta = (\eta, \gamma) \in \R^p$ can always be constructed satisfying
$$X_A \eta + P_A X_{A^c} \gamma = 0,$$
since $P_A X_{A^c} \gamma \in \spn(X_A)$.
Define a mapping $g:i \mapsto i + |A|$ for $i\in[p]$ and let $B = [|A|] \cup g(D) \subset [p]$.
Then $|B|=|A|+|D|\leq s$, and $\delta \in \scrC(B, 3)$ because
\begin{align*}
\sum_{i \in B^c} |\delta_i| = \sum_{i \in D^c} |\gamma_i| \leq 3 \sum_{i \in D} |\gamma_i| \leq 3 \sum_{i \in B} |\delta_i|,
\end{align*}
where the second step is due to $\gamma\in\scrC(D, 3)$.
Based on that $X$ satisfies $\textup{RE}(s,3)$, we arrive at the following inequality:
\begin{align*}
\frac{\|V^\trans X_{A^c}\gamma \|}{{\sqrt{n-k}}\|\gamma_D\|} &= \frac{\| X_A\eta + P_A X_{A^c}\gamma + P_A^\perp X_{A^c}\gamma \|}{\sqrt{n-k}\|\gamma_D\|} \\
& =  \frac{\sqrt{n}}{\sqrt{n-k}}\frac{\| X\delta\|}{\sqn\|\gamma_D\|}
\geq  \frac{\sqrt{n}}{\sqrt{n-k}}\frac{\| X\delta\|}{\sqn\| \delta_B\|}
\geq \frac{\sqrt{n}}{\sqrt{n-k}} \kappa(s, 3; X),
\end{align*}
which shows that $\textup{RE}(s-k, 3)$ holds for $V^\trans X_{A^c}$ and
$\kappa(s-k, 3; V^\trans X_{A^c}) \geq \sqrt{n/(n-k)}\kappa(s, 3; X) $.
Lastly,
$n\lmd_2 = K\sigma\sqrt{n\log(p-k)}\geq K\sigma\sqrt{(n-k)\log(p-k)}$, as required in Lemma~\ref{lemma:LassoLoss}.

So far, we have shown that $(V^\trans X_{A^c}, V^\trans y)$ and $\lmd_2$ satisfy all the conditions in Lemma~\ref{lemma:LassoLoss}, which with \eqref{eq:equiprloss} implies that
\begin{align*}
& \Prob\left\{\| P_A^\perp X (\tdbeta-\beta)\|^2
\leq\frac{16nK^2\sigma^2 \omega(V^\trans X_{A^c})}{(n-k) \kappa^2(s-k,3;V^\trans X_{A^c})}(s - k)\log (p - k)\right\} \\
& \quad \geq 1-(p-k)^{1-K^2/8},
\end{align*}
for any $A \subseteq \supp(\beta)$. Then inequality \eqref{eq:Lassoperp} immediately follows by noting that $\omega(V^\trans X_{A^c}) \leq \omega(X)$ and substituting $\kappa(s-k,3;V^\trans X_{A^c})$ with
$\sqrt{n/(n-k)}\kappa(s, 3; X)$.
\end{proof}

\bibliographystyle{asa}
\bibliography{sparsereferences}

\end{document}